\documentclass[runningheads]{llncs}
\usepackage[utf8]{inputenc}
\usepackage{multicol}

\usepackage{hyperref}
\usepackage{amssymb}
\usepackage{svg}
\usepackage{xfrac,nicefrac}
\usepackage{scalerel}
\usepackage{graphicx}
\usepackage[caption=false]{subfig}
\usepackage{float}
\usepackage{amsmath}
\usepackage{mathtools}
\usepackage{cleveref}
\usepackage{xcolor}
\usepackage{graphicx}
\usepackage{algorithm}
\usepackage{algpseudocode}
\usepackage{todonotes}
\usepackage{breakcites}

\usepackage[export]{adjustbox}

\newcommand{\SubV}[0]{{\scaleto{V}{2.5pt}}}
\newcommand{\Locs}[0]{S_{M\times A}}

\newcommand{\Img}[0]{\mathrm{Im}}
\newcommand{\Lcal}[0]{\mathcal{L}}
\newcommand{\LA}[0]{{\mathcal{L}_A}}

\newcommand{\N}[0]{\mathbb{N}}

\newcommand{\Out}[0]{\mathrm{Out}}

\newcommand{\emptyword}{\epsilon}
\newcommand{\divr}[0]{\ \#_E \ }

\newcommand{\simeqq}[0]{\mathrel{\substack{
			\textstyle\sim\\[-0.4ex]
			\textstyle=}}}

\spnewtheorem{obs}{Observation}{\bfseries}{\itshape}

\title{Conformance Testing of Mealy Machines Under Input Restrictions}
\author{Alberto Larrauri, Roderick Bloem}
\institute{TU Graz, TU Graz}
\date{July 2021}

\begin{document}
	
	\maketitle
	
	\begin{abstract}
	We introduce a grey-box conformance testing method for networks of interconnected Mealy Machines. This approach addresses the scenario where all interfaces of the component under test are observable, but its inputs are under the control of other white-box components. We prove new
	conditions for full fault detection that
	exploit repetitions across branching
	executions of the composite machine in a novel way. Finally, we provide
	experimental evaluation of our
	approach on cascade compositions of 
	up to a thousand states, and show
	that it notably outperforms existing black-box testing techniques. 
	\end{abstract}
	
	\section{Introduction/motivation}
	
	In this paper we propose a grey-box testing approach for networks of interconnected Mealy Machines. We address the scenario where 
	all communications of the component under test can be observed, but some of its inputs are controlled by other white-box parts of the system. The presented method falls within the scope of conformance testing of finite state machines (FSMs)
	\cite{broyModelbasedTestingReactive2005, dorofeevaFSMbasedConformanceTesting2010a}
	\par
	In its most studied variant, the conformance testing problem for FSMs
	deals with deterministic and input-complete FSMs, i.e., Mealy machines 
	\cite{mooreGedankenExperimentsSequentialMachines1956,t.s.chowTestingSoftwareDesign1978,simaoReducingTestLength2012a,fujiwaraTestSelectionBased1991,dorofeevaImprovedConformanceTesting2005,souchaSPYHMethodImprovementTesting2018}.
	In this setting, we consider a 
	fully known Mealy machine $M$ (the specification) and a black box $B$, for which
	we only know a bound $k$ on
	the number of states. The goal is to design a test suite to determine whether the black box $B$ conforms (is equivalent) to the  $M$. \par
	FSM-based conformance testing is an active
	research area and numerous techniques exist in the literature (see the survey \cite{dorofeevaFSMbasedConformanceTesting2010a},
	or\cite{souzaHSwitchCoverNew2017}). The primary
	motivation of these techniques is the verification of reactive systems for which FSMs are a suitable model. Despite its simplicity, the FSM formalism is used in very diverse domains,
	yielding a broad range of applications for FSM-based testing
	\cite{broyModelbasedTestingReactive2005}.
	Another notable application of conformance testing lies 
	in automata learning
	\cite{delahigueraGrammaticalInferenceLearning2010} and derived procedures, such as black-box checking \cite{peledBlackBoxChecking1999}. 
	In the setting of the ``minimally adequate teacher'' introduced by Angluin \cite{angluinLearningRegularSets1987}, such techniques require an equivalence oracle in their application.  
	However, these oracles are largely impossible to obtain when dealing	with black box systems. Thus, in practice 
	equivalence queries are simulated through various kinds of testing strategies
	\cite{learnlib}. Furthermore, 
	there is a well-known close relation between model inference
	and conformance testing (see  \cite{bergCorrespondenceConformanceTesting2005}) that
	extends to even more recent automata learning techniques
	that do not require equivalence oracles \cite{souchaObservationTreeApproach2020,vaandragerNewApproachActive2021a}. \par
	
	In reality however, reactive 
	systems rarely consist of a single monolithic structure, but instead consist of smaller interacting components.
	Existing FSM-based techniques developed for black-box systems
	are not fit to deal with this context, as they suffer
	from the problem of state explosion and rapidly hit a wall.   
	Hence, there is a need for gray-box methods
	able to exploit information about known internal
	components and their communications. There are a few 
	notable conformance testing works in this direction
	\cite{petrenkoLearningCommunicatingState2019a,petrenkoTestingStrategiesCommunicating1995, petrenkoTestingPartialDeterministic2005},
	but it remains a relatively unexplored area. \par

	We consider a scenario where all interfaces of
	the component under test $B$ are observable,
	but its inputs are controlled by other known components
	of the system. The prototypical example of this occurs when $B$ is the tail component of a cascade composition of Mealy machines $B\circ H$, as depicted in Figure~\ref{fig:cascade}.
	The State-Counting method
	\cite{petrenkoTestingPartialDeterministic2005},
	one of the main approaches for this situation,	 
	resorts to treating $B$ as a partially specified
	Mealy machine- i.e., a machine where some transitions 
	are missing. This reduction relies on a classical
	construction for component minimization by Kim and 
	Newborn \cite{joonkikimSimplificationSequentialMachines1972}
	which involves an exponential blow-up of the problem's size.
	 However, it has been shown recently that this expensive 
	construction is not required
	to optimize components \cite{larrauriMinimizationSynthesisTail2021}, and that
	cheaper techniques may be used instead. \par
	Our main contribution in this paper is 
	a generalization of the State-Counting method which 
	avoids the Kim-Newborn construction. In order to achieve
	this, we develop a formalism for reasoning about
	interleaving executions in systems with universal
	branching. This allows us to prove new sufficient conditions
	for complete fault-detection in the gray-box setting. 
	We give two testing algorithms making use of this newly
	introduced theory, and show experimentally that
	they are able to handle compositions of up to a thousand states,
	whereas experimental data on reasonably sized examples does not exist for the state-of-the-art \cite{petrenkoLearningCommunicatingState2019a,petrenkoTestingPartialDeterministic2005}. 
	Additionally, we show a practical relation between the gray-box testing task
	and the classical problem of determining language inclusion between
	non deterministic automata (NFA) \cite{kupfermanVerificationFairTransition1996}, as well as the 
	problem of state reduction for NFAs \cite{jiangMinimalNFAProblems1993}. 

	\section{Preliminaries}
	\paragraph*{General Notation}
	Given an alphabet $X$, 
	we write $X^*$
	for the set of finite words of arbitrary length over $X$.
	We use $\emptyword$ to denote the empty word, and 
	given a word $\alpha$, $|\alpha|$ stands for its length.
	We write ($\alpha< \beta$) $\alpha\leq \beta$  when $\alpha$ is a (strict) prefix of $\beta$.

	\paragraph*{Automata Over Finite Words}
	We consider automata over finite words
	where all states 
	are accepting. Let $\varphi$ be a finite alphabet. A \textbf{non-deterministic finite automaton (NFA)} $A$ \textbf{over} $\varphi$,
	is a tuple $(\varphi, S_A, \Delta_A, r_A)$, 
	where $S_A$ is a finite set of states, $\Delta_A: S_A\times \varphi \rightarrow 2^{S_A}$ 
	is the transition function, and $r_A\in S_A$ is the 
	initial state.  A \textbf{run} of $A$ on a word $\alpha\in \varphi^*$ is defined as usual. We say that an state $s\in S_A$ \textbf{accepts} a word $\alpha$
	if there is a run of $A$ on $\alpha$ starting from $s$. If $s=r_A$
	we simply say that $A$ accepts $\alpha$.
	The \textbf{language} of 
	$s$ is the set $\LA(s)\subseteq \varphi^*$ containing the words accepted by $s$. Note that $\LA(s)$
	is prefix-closed. We simply write $\LA$ for $\LA(r_A)$. 
	We lift $\Delta_A$ 
	to words $\alpha\in \varphi^*$ 
	in the natural way. The set
	$\Delta_A(s,\alpha)$ consists
	of all $s^\prime$ such that some run of $A$ on $\alpha$
	from $s$ finishes at $s^\prime$.
	We write $\Delta_A(\alpha)$ for $\Delta_A(r_A,\alpha)$.

	\paragraph*{Mealy Machines}
	\label{sec:mealy}
	A \textbf{Mealy machine} $M$ is a tuple 
	$(I_M, O_M, S_M,\delta_M,$ $\lambda_M, r_M)$,
	where $I_M, O_M$ are finite alphabets,
	$S_M$ 
	is a finite set of states,
	$\delta_M: S_M\times I_M \rightarrow S_M$ 
	is the next state function, $\lambda_M: S_M\times I_M \rightarrow O_M$ 
	is the output function and $r_M\in S_M$ is the initial 
	state.
	We lift $\delta_M$ and $\lambda_M$
	to input sequences in the natural way.
	We define $\delta_M(s, \emptyword)=s$,
	$\lambda_M(s,\emptyword)=\emptyword$ for all $s$. 
	Given $\alpha \in I_M^*, x \in I_M$, if 
	$s^\prime = \delta_M(s, \alpha)$,
	then
	$\delta_M(s,\alpha x)=\delta_M(s^\prime,x)$ and
	$\lambda_M(s,\alpha x) = 
	\lambda_M(s,\alpha)\lambda_M(s^\prime,x)$. We write $\delta_M(\alpha)$ and $\lambda_M(\alpha)$
	for $\delta_M(r_M,\alpha)$ and $\lambda_M(r_M,\alpha)$ 
	respectively. We say that $M$
	is \textbf{reduced} if for any pair of different states $s_1,s_2\in S_M$ there is a word
	$\alpha\in I_M$ distinguishing them, i.e., $\lambda_M(s_1,\alpha)\neq
	\lambda_M(s_2,\alpha)$.
	We define $\Out(M)$
	as the set of words $\lambda_M(\alpha)$, for
	all $\alpha\in I_M^*$.\par

	
	
	\subsection{Conformance Testing}
	\label{sec:classic_testing}

	Let $M$ be a Mealy machine 
	representing an intended model or \textbf{specification} for
	a \textbf{black-box system} $B$. 
	A \textbf{test suite for} $M$
	is a finite prefix-closed set 
	$E \subseteq (I_M)^*$. Sequences $\alpha\in (I_M)^*$ are
	called \textbf{tests}. We define suites as prefix-closed sets, because it simplifies the exposition of technical results later on. However, in practice only the maximal tests in a suite $E$ are relevant. This is because once the output response $\lambda_B(\alpha)$
	of $B$ to a test $\alpha$ is observed, 
	the outputs $\lambda_B(\beta)$ for all $\beta\leq \alpha$ are known as well. Thus, we define the \textbf{total length}, or the
	\textbf{number of symbols} of a suite $E$ as the sum of the lengths of its maximal tests. \par 
	We denote by $\Im$ 
	the set of Mealy machines 
	$N$ with the same input/output alphabets as $M$, and write $\Im_k$
	for the set of those with at most $k$ states. Given a machine $N\in \Im$, and a set $V\subseteq (I_M)^*$, we write $M\sim_V N$ if $\lambda_M(\alpha)=\lambda_N(\alpha)$
	for all $\alpha\in V$, or simply write
	$M\sim N$ when $V=(I_M)^*$.  We say that a suite $E$
	is $k$-complete if
	$M\sim_E N$ implies $M\sim N$
	for all $N\in \Im_k$. The conformance testing
	problem for Mealy machines is as follows. 
	
	\begin{problem}[Unrestricted conformance testing] \label{prob:classic} Given a Mealy machine $M$ and a number $k\in \N$, 
		compute a $k$-complete suite $E$ for $M$.
	\end{problem}
	
	There are three main parameters to optimize in this problem: 
	running time, number of maximal tests in the suite $E$, and number of symbols. The last two objectives are important because a suite may be used on multiple black boxes after its construction, or these black-box systems may be slow to execute. Thus, for some applications it may be worthwhile to develop a slower algorithm that results in smaller suites. We adopt the convention that suites produced by conformance testing algorithms are returned by listing their maximal tests. In these circumstances, the time cost of such algorithms is trivially bounded by the total length of the suites they construct.
	\par
	
	Methods developed to solve Problem~\ref{prob:classic} can be understood as modifications of the first technique, the W-method \cite{t.s.chowTestingSoftwareDesign1978,vasilevskiiFailureDiagnosisAutomata1975}. Despite the notable
	experimental improvements (e.g., \cite{dorofeevaFSMbasedConformanceTesting2010a,souchaSPYHMethodImprovementTesting2018}),
	the worst-case analysis of newer techniques does not improve that of the original algorithm, as the W-method is optimal in the worst case \cite{vasilevskiiFailureDiagnosisAutomata1975}.
	\par
	We discuss now the complexity of the W-method. 	
	Fix a reduced specification machine $M$. We 
	call the parameter $e\coloneqq k-|S_M|$ the number of
	\textbf{extra states}. This is a central variable in
	conformance testing, as it measures the uncertainty 
	about the black-box under test. The
	problem only is meaningful when $e\geq 0$.
	The number of (maximal) tests produced by the W-method is $O(|S_M|^2|I_M|^{e + 1})$, and the total number of symbols, as well as its time cost, are given by $O(|S_M|^2 k |I_M|^{e + 1})$. 
	Some insight on these bounds can be gained
	from the general structure of conformance testing
	methods. In most of them, the suite $E$
	is built in three stages. First, 
	one constructs a state-cover $V$ of $M$- i.e.,
	a set containing a word $\alpha$ with
	$\delta_M(\alpha)=s$ for each $s\in S_M$. Afterwards, one appends to $V$ the so-called
	\textbf{traversal set} $(I_M)^{e+1}$, of arbitrary words	of length $e+1$. This addition is unavoidable and it is responsible for the exponential factor in the previous bounds. Finally, some distinguishing suffixes are
	appended to each word in $V\cdot(I_M)^{e+1}$. 
    Improvements over the W-method usually revolve around modifications of this last step.

	\section{Problem Statement}

	In this section we introduce the 
	\textbf{restricted conformance testing} problem, which is the main subject of this text. As before, let $M$ be a Mealy machine representing a specification for a black box $B$. Let $A$ be an NFA over $I_M$ representing the \textbf{context} in which $B$ operates. 
	We consider the extension of the
	conformance testing problem where it is not possible to apply arbitrary tests to $B$, but only those sequences in $\LA$ can be used instead. Furthermore, now we do not ask whether $M$ and $B$
	are equivalent, but just whether they respond equally to sequences in $\LA$. That is, whether $M\sim_\LA B$. \par A \textbf{test suite for} $M$ \textbf{in the context of} $\LA$ is a finite prefix-closed set $E\subseteq \LA$. Analogously to before, we say that $E$ is
	$k$-\textbf{complete} (in the context of $\LA$) if whenever $M\sim_E N$ for some $N\in \Im_k$, it also holds that
	$M\sim_\LA N$. Sometimes we will drop the phrase ``in the context of $\LA$'', and simply say that $E$ is $k$-complete when $\LA$ is implied and there is no ambiguity. 
	We study the following problem: \par
	
	\begin{problem}[Restricted conformance testing] \label{prob:main}
		Provided with a Mealy machine $M$, an NFA $A$ over $I_M$, and some $k\in \N$, compute a $k$-complete suite $E\subseteq \LA$ for $M$ in the context of $\LA$.
	\end{problem}
	
		\begin{figure}
			\centering
	\includegraphics[width=0.5\textwidth]{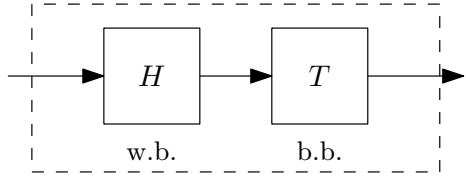}
	\label{fig:cascade}
	\caption{A cascade composition of Mealy machines}
	\end{figure}
	
	As mentioned during the introduction, our motivation for this task lies in the gray-box testing problem were the component under test has
	observable interfaces, but uncontrollable inputs.
	During this paper, we focus in the following
	particular case.

	\subsection{Testing of The Tail Component}
	\label{sec:tail}
	A \textbf{cascade composition} $T\circ H$ of two Mealy machines, $T$ and $H$ consists in a 
	one-way sequential connection of both, where
	the head $H$ processes external inputs and the tail $T$ reacts to $H$'s outputs (Figure~\ref{fig:cascade}). 
	In this setting, $T$ can only respond to sequences belonging to $\Out(H)$. An NFA
	representing this language is easily obtained by ``removing'' the input symbols from $H$'s transitions, as shown in Figure~\ref{fig:img} \cite{joonkikimSimplificationSequentialMachines1972}. This is called the \textbf{image automaton of $H$}, $\Img(H)$. This construction
	shows a straight-forward reduction of the following task to Problem~\ref{prob:main}:
	
	\begin{problem}[Tail component testing] \label{prob:tail}
		Given a cascade of Mealy machines $T\circ H$, and some $k\in \N$, compute a $k$-complete suite $E\subseteq \Out(H)$ for $T$ in the context of $\Out(H)$.
	\end{problem}
	\vspace*{-15pt}
	
		\begin{figure}[ht]
			\centering
			\hfill
			\subfloat[A Mealy machine $H$.
			\label{fig:H}]{
			\includegraphics[width=0.37\textwidth]{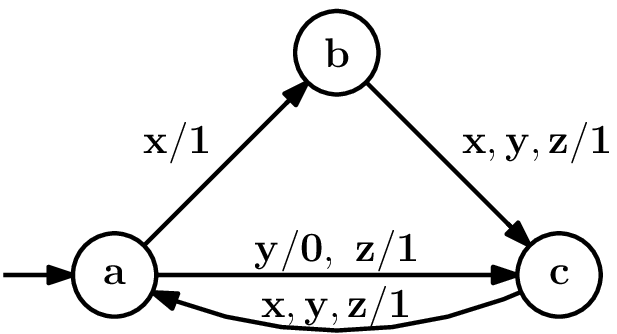}
			}
			\hfill
			\subfloat[An NFA $A$ accepting $\Out(H)$.\label{fig:A}]{
			\includegraphics[width=0.37\textwidth]{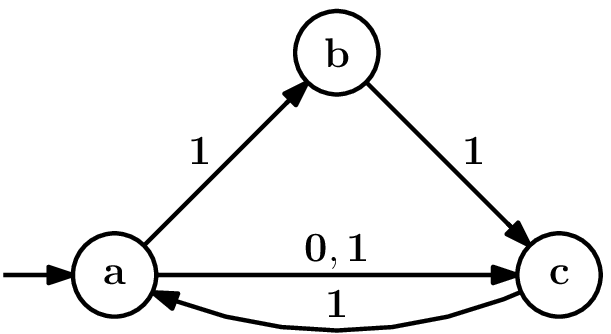}
			}
			\hfill
			\hfill
		
			\vspace{5mm}
			\caption{Construction of the image
			automaton for a Mealy machine.}
			\label{fig:img}

		\end{figure}

	To simplify the discussion we will use this particular case of component testing
	to motivate our main problem. However, more general forms of component testing
	were interfaces
	are observable
	can also be addressed via Problem~\ref{prob:main},
	as there are
	polynomial reductions 
	transforming this scenarios
	into cascade compositions \cite{wangInputDonCare1993,larrauriMinimizationSynthesisTail2021}. Now we give a brief overview existing solutions for the Tail Testing problem. \par
	
	\subsubsection{Baseline Solution: Testing of The Composite Machine}
	Given a cascade $T\circ H$, and a bound $k$,
	one can use existing black-box testing methods
	to solve Problem~\ref{prob:tail} in the following
	way. First, a Mealy machine $P$
	representing the whole composition
	can be obtained via a simple product construction \cite{harrisSynthesisFiniteState1998}.
	Here, $I_P=I_H$, $O_P=O_H\times O_T$, and 
	$S_P\subseteq S_H\times S_T$. Afterwards, 
	one can apply any existing conformance testing method to obtain a $|S_H|k$-complete suite $E$ for	$P$. Finally, computing the image $\lambda_H(E)$ of $E$ through $H$ we obtain a 
	$k$-complete suite for $T$ in the context of $\Out(H)$. Taking into account the bounds in
	\cref{sec:classic_testing}, the complexity of this approach is $O(|S_H|^3|S_T|^2 k |I_M|^{e_P + 1})$, where $e_P\coloneqq k|S_H| - |S_P|$. We note that even when
	$k=|S_T|$ and the original problem presents
	no extra states,
	$|S_P|$ can be much smaller than $|S_H||S_T|$,
	yielding a large $e_P$ and making this approach impractical. We refer to this problem as 
	the \textbf{blow-up of extra states}.\par

\subsubsection{Related Work} To the date
	there are two main approaches proposed
	for the Tail Testing problem which
	aim to overcome the blow-up of extra states
	of the previous method. They are the State-Counting method \cite{petrenkoTestingPartialDeterministic2005} and a more
	recent SAT-based technique \cite{petrenkoLearningCommunicatingState2019a}.
	Each one of these techniques encounter
	important issues in their complexity analyses, however, 
	and there is a lack of experimental data
	about their performance outside very small 
	examples (compositions not reaching ten
	states in total). \par
	The State-Counting method \cite{petrenkoTestingPartialDeterministic2005}
	gives sufficient conditions for complete fault detection in presence of input restrictions. 
	In order to apply these conditions to Problem~\ref{prob:tail}, one has to employ 
	the Kim-Newborn construction
	\cite{joonkikimSimplificationSequentialMachines1972},
	as described in \cite{petrenkoTestingStrategiesCommunicating1995}. This involves constructing a so-called ``incompletely specified machine'' $P^\prime$, via a product of $T$ and the determinization of the image automaton $\Img(H)$. The resulting size
	of $P^\prime$ is $|S_T|2^{|S_H|}$ in the worst case. This machine $P^\prime$ is used later
	as the specification model to produce a
	$k$-complete
	suite. The drawback of this analysis is, however,
	that this model $P^\prime$ can be exponentially
	bigger than the composite machine $P$ in the baseline method. This potentially yields exponentially larger suites with exponentially longer tests. \par
	The SAT-based approach in \cite{petrenkoLearningCommunicatingState2019a} constructs a $k$-complete suite $E$ for	$T$ in an iterative way, 
	asking a SAT solver whether there is
	some $T^\prime\in \Im_k$ with
	$T\sim_E T^\prime$ but
	$T\nsim_{\Out(H)} T^\prime$. If the answer is negative, $E$ is already $k$-complete.
	Otherwise, a suitable distinguishing sequence for $T$ and $T^\prime$ is added to $E$. This technique has the potential for producing small suites, but the
	drawback of having to perform a possibly expensive SAT
	call for the computation of each individual test, whose cost scales exponentially 
	with $|S_H|, |S_T|, k$ and $|E|$.\par
	
	\section{Theoretical Analysis}
    \label{sec:theory}
	During this section $M$ denotes a specification Mealy machine, $A$ a context NFA over $I_M$, and $E$ an unspecified test suite $E\subseteq \LA$. Lastly, we consider a reflexive binary
	relation $\sqsubseteq$ over $S_A$ which under-approximates language containment. 
	That is, $a \sqsubseteq b$ implies $\LA(a)\subseteq
	\LA(b)$ for all $a,b\in S_A$. The goal of this section is to give sufficient conditions for $k$-completeness of the suite $E$. These, in turn, will provide the correctness guarantees for our proposed algorithms (Section~\ref{sec:algs}). \par

    Our sufficient conditions build upon those in the State-Counting method \cite{petrenkoTestingPartialDeterministic2005}, and can be seen as a generalization
    of them. Informally, the main difference
    is that the State-Counting method only relates to the case where $A$ is deterministic.

	\subsection{Product of a Mealy Machine with an NFA}
	
	Suppose we want to study the observable 
	behaviours of $M$ after the application
	of a test $\alpha\in \LA$.
	Here, not only is it relevant to know the state $\delta_M(\alpha)$, but also the set
	of possible context states
	$a\in \Delta_A(\alpha)$. This is 
	because these states $a$
	determine which suffixes that can extend the test $\alpha$. Thus, in our setting,
	state pairs $(s,a)\in S_M\times 
	S_A$ play a major role. \par
	
	The \textbf{product transition function } is the
	map given by $
	\Delta_{M\times A}( (s,a), \alpha)=
	\{\delta_M(s,\alpha)\} \times \Delta_A(a,\alpha)$, for any
	$(s,a)\in S_M\times S_A$, $\alpha\in I_M^*$. Additionally, given $\alpha\in I_M^*$, we write $\Delta_{S\times A}(\alpha)$ to denote
	$\Delta_{S\times A}((r_M,r_A), \alpha)$.

	\begin{figure}[ht]
			\centering
			\hfill
			\begin{minipage}[b]{0.30\textwidth}
			\subfloat[A Mealy machine $M$.
			\label{fig:M}]{
			\includegraphics[width=\textwidth]{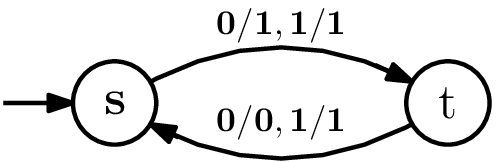}
			}\\
			\centering
			\subfloat[An NFA $A$ over $I_M$.
			\label{fig:A2}]{
			\includegraphics[width=\textwidth]{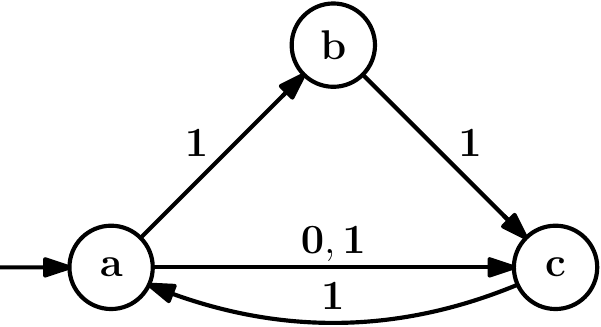}
			}
			\end{minipage}
			\hfill
			\begin{minipage}[b]{0.45\textwidth}
			\subfloat[The product of $M$ and $A$.\label{fig:prod}]{
			\includegraphics[width=\textwidth]{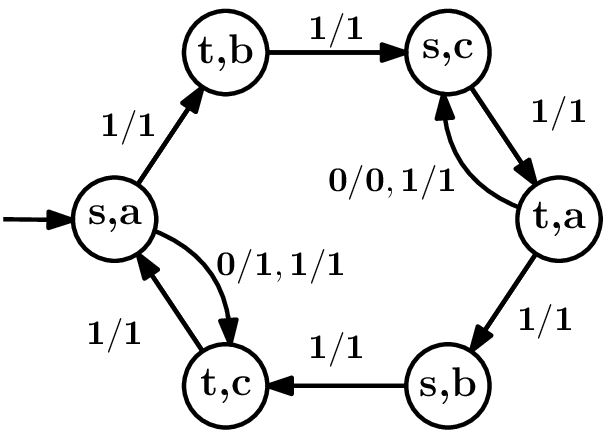}
			}
			\end{minipage}
			\hfill\hfill

			\caption{Representation of the product of a Mealy machine and a context NFA.}
			\label{fig:locs}
			
		\end{figure}
	
    Informally,
	the semantics of the product $M\times A$ equipped with
	$\Delta_{M\times A}$ are those of a universally branching machine. Given an input word $\alpha$, an execution of this product
	consists on multiple parallel runs, each one being the product
	of a single run of $A$ on $\alpha$
	with the deterministic run of $M$ on this sequence. This notion of product of a Mealy machine
	with an NFA is explored in greater detail in \cite{larrauriMinimizationSynthesisTail2021}. \par
	
	In a state pair $(s,a)\in S_M\times S_A$, the state
	$s$ of $M$ is responsible for the input/output behaviour,
	while $a$ represents the input sequences 
	that are non-blocking at this point. 
	Two pairs $(s,a), (t,b)\in S_M\times S_A$ are \textbf{distinguishable} or \textbf{incompatible}, denoted $(s,a)\nsim (t,b)$
	if $\lambda_M(s,\alpha)\neq \lambda_M(t,\alpha)$
	for some sequence $\alpha$ available in both $(s,a)$
	and $(t,b)$, i.e., $\alpha\in \LA(a)\cap \LA(b)$.
	In this situation we say that
	$\alpha$ witnesses $(s,a)\nsim (t,b)$, written
	$\alpha \models (s,a)\nsim (t,b)$. 
	Two pairs are \textbf{equivalent}, denoted $(s,a)\cong (t,b)$, if, in addition to being compatible, it holds $a=b$. 
	\par
    In the following result, we bound the 
	length of shortest distinguishing sequences
	for state-pairs in $S_M\times S_A$ (see Appendix~\ref{ap:seq} for the proof):
	
	\begin{theorem} \label{thm:seq}
		Let $M$ be a Mealy machine and let 
		$A$ be an NFA over $I_M$. Let $s,t\in S_M$ $b,a\in S_A$.
		Suppose that $(s,a)\nsim (t,b)$
		as well as
		$a\sqsupseteq b$. Then 
		there exists some $\alpha\in 
		\LA(b)$ satisfying both
		$\alpha\models (s,a)\nsim (t,b)$
		and $|\alpha|\leq |S_M||S_A|$.
	\end{theorem}

	

	\subsection{Context Tree}
	
	During our discussions we need to consider the 
	``unrolling'' of the context automaton $A$ on 
	various words. We formalize this notion in the following definition. 
	The \textbf{context tree} is the set
	$\Gamma\subseteq \LA \times S_A$ consisting
	of the pairs $\sfrac{a}{\alpha}$, where $a\in 
	\Delta_A(\alpha)$. The elements $\sfrac{a}{\alpha}$
	of the testing tree are called \textbf{nodes}. A node
	$\sfrac{a}{\alpha}$ is read as ``$a$ at $\alpha$'',
	and represents a point during an execution of $A$.
	Given a set of sequences $D\subseteq \LA$, we put
	$\Gamma(D)$ for the nodes $\sfrac{a}{\alpha}\in \Gamma$ with $\alpha\in D$. We say that a node
	$\sfrac{b}{\beta}$ \textbf{precedes} another one $\sfrac{a}{\alpha}$, written $\sfrac{b}{\beta}\preceq
	\sfrac{a}{\alpha}$, if $\alpha=\beta \gamma$ 
	and $a\in \Delta_A(b,\gamma)$, for some $\gamma$. \par


	Two tests $\alpha, \beta \in E$ are called $E$-\textbf{separable}, denoted $\alpha \ \#_E \ \beta$,
	if there is a suffix $\gamma$ 
	satisfying $\alpha\gamma, \beta\gamma\in E$
	and $\lambda_M(\delta_M(\alpha),\gamma)\neq
	\lambda_M(\delta_M(\beta),\gamma)$. 
	This notion of separability has been used in classical conformance testing \cite{simaoReducingTestLength2012a}, and learning (called ``apartness'') \cite{vaandragerNewApproachActive2021a}. 
	The following result gives justification for it. 
	\begin{lemma} \label{lem:aux1}
		Suppose that $\alpha\divr \beta$ for two tests $\alpha,\beta\in E$. Then 
		$\delta_N(\alpha)\neq \delta_N(\beta)$
		for any $N\in \Im$ satisfying $M\sim_E N$.
	\end{lemma}
	
	Each node $\sfrac{a}{\alpha}\in \Gamma$ corresponds
	naturally to a location $(\delta_M(\alpha), a)$ $\in \Locs$. 
	Given a set
	of nodes $R\subseteq \Gamma(E)$,
	we say that $E$ is \textbf{incompatibility-preserving
	with respect to} (w.r.t.) $R$
	if for any $\sfrac{a}{\alpha}, \sfrac{b}{\beta}\in C$
	with $(\delta_M(\alpha),a)\nsim (\delta_M(\beta), b)$
	it holds $\alpha \divr \beta$.

	\subsection{Rankings and Basic Proof of Completeness}
	
	During this section we prove a weaker version of our main result where the central arguments of the full proof
	are showcased. 
	A \textbf{node ranking} is a sequence
	$(\sfrac{a_j}{\alpha_j})_{j=1}^m \subseteq \Gamma$ 
	of nodes 
	where $\alpha_1 < \dots < \alpha_m$. We call a ranking \textbf{flat} if 
	$a_1=\dots = a_m$, and 
	\textbf{monotonous} 
	if $a_1 \sqsupseteq \dots \sqsupseteq a_m$. 
	We will be loose with the use of notation and treat 
	rankings as sets when convenient, instead of sequences. We write $R\preceq \sfrac{a}{\alpha}$ for a ranking $R$ whenever $\sfrac{b}{\beta}\preceq \sfrac{a}{\alpha}$ holds for all elements
$\sfrac{b}{\beta} \in R$.\par
	We say that a node $\sfrac{a}{\alpha}$ is $k$\textbf{-saturated} if there is a monotonous ranking $R\subseteq \Gamma(E)$ with $|R|=k$, where $\sfrac{b}{\beta}\preceq \sfrac{a}{\alpha}$ for all $\sfrac{b}{\beta}\in C$,
	and $E$ is incompatibility-preserving w.r.t. $R$. If all the nodes $\sfrac{a^\prime}{\alpha}$ with $a^\prime \in \Delta_A(\alpha)$ are $k$-saturated, then we say that
	the sequence $\alpha$ is $k$-\textbf{saturated} itself.

	\begin{theorem}
		\label{thm:main_no_cover}
		Suppose that all tests $\alpha\in \LA \setminus E$ have a prefix $\beta\in E$
		which is $(k+1)$-saturated. Then $E$
		is $k$-complete. 
	\end{theorem}

	\begin{proof}
	    The proof follows an argument
	    of infinite descent. The idea
	    is that given a test
	    $\alpha\in \LA \setminus E$ which detects
	    a fault not covered by $E$, another
	    strictly shorter sequence $\alpha^\prime$ with the same properties can be found. As decreasing sequences
	    of natural numbers are necessarily finite, this scenario is impossible and
	    full fault detection by $E$ is guaranteed. The central part of the ``shrinking'' argument is that
	    whenever a sufficiently large ranking
	    $R$ can be found throughout a test $\alpha$,
	    then this sequence necessarily follows a
	    ``lasso''-like path
	    in the product $M\times A$
	    and some central portion of $\alpha$
	    can be removed.\par
		We proceed by contradiction. 
		Let $N\in \Im_k$ be a machine satisfying both $M\sim_E N$ 
		and $M\nsim_\LA N$. Let $\alpha\in \LA \setminus E$ be a shortest test
		distinguishing $M$ and $N$. We show that it is possible to build an even shorter sequence $\alpha^\prime$ that
		also distinguishes $M$ and $N$.  
		Let $\beta\in E$ be a $(k+1)$-saturated prefix of $\alpha$, and let $\gamma$
		be the suffix satisfying $\beta\gamma=\alpha$ 
		. 
		As $\alpha\in \LA$, it must be that
		$\gamma\in \LA(b)$ for some $b\in \Delta_A(\beta)$. The node $\sfrac{b}{\beta}$
		is $(k+1)$-saturated, so there is
		some monotonous ranking $R
		\subseteq \Gamma(E)$ witnessing this property. Let $R=(\sfrac{c_j}{\varphi_j})_{j=1}^{k+1}$. As $|S_N|\leq k$, by the pigeonhole principle there must be two indices
		$x< y$ for which $\delta_N(\varphi_{x})=\delta_N(\varphi_{y})$. Let $\omega$ be the suffix satisfying $\varphi_{y}\omega=\beta$,
		and $\varphi_{y}\omega\gamma=\alpha$.
		Let $\alpha^\prime\coloneqq\varphi_{x}\omega\gamma$ 
		The following statements hold true: \par
			

		\begin{claim}[I] $\alpha^\prime \in \LA$.
		\end{claim}
		First, note that $\omega\gamma\in \LA(c_{y})$. Indeed, this follows from $(\varphi_{y}, c_{y})\preceq 
		\sfrac{b}{\beta}$ together with $\gamma\in 
		\LA(b)$. As $c_{x}\sqsupseteq c_{y}$, it also holds that 
		$\omega\gamma\in \LA(c_{x})$. This,
		in conjunction with $c_{x}\in \Delta_A(\varphi_{x})$, shows the claim. 
		\begin{claim}
		$\lambda_M(\delta_M(\varphi_{y}),\omega\gamma)\neq \lambda_N(\delta_N(\varphi_{y}),\omega\gamma)$.
		\end{claim}
		The fact that $M\sim_E N$ and $\varphi_{y}\in E$, implies $\lambda_M(\varphi_{y})=\lambda_N(\varphi_{y})$. However, we know that
		$\lambda_M(\alpha)\neq \lambda_N(\alpha)$, and $\alpha= \varphi_{y}\omega\gamma$, so the claim follows. \par
		\begin{claim}[III]
			$\lambda_N(\delta_N(\varphi_{x}),\omega\gamma)= \lambda_N(\delta_N(\varphi_{y}),\omega\gamma)$.
		\end{claim}
		This is straight-forward, as
		$\delta_N(\varphi_{x})= \delta_N(\varphi_{y})$.\par
		
		\begin{claim}[IV]
			$\lambda_M(\delta_M(\varphi_{x}),\omega\gamma)= \lambda_M(\delta_M(\varphi_{y}),\omega\gamma)$.
		\end{claim}
		Suppose that $(\delta_M(\varphi_{x}),c_{x})\nsim 
		(\delta_M(\varphi_{y}),c_{y})$. 
		As $R$
		is a ranking witnessing that $\sfrac{b}{\beta}$ is $(k+1)$-saturated, $E$ is
		incompatibility-preserving w.r.t. $R$. Thus, 
		$\varphi_{x}\divr \varphi_{y}$ follows. However, by Lemma~\ref{lem:aux1} this
		contradicts the fact that $\delta_N(\varphi_{x})=
		\delta_N(\varphi_{y})$ while at the same time $M\sim_E N$. Hence,
		$(\delta_M(\varphi_{x}),c_{x})\sim 
		(\delta_M(\varphi_{y}),c_{y})$ must hold. This implies the statement, because 
		$\omega\gamma \in \LA(c_{x}) \cap 
		\LA(c_{y})$, as evidenced during the first claim. \\~\\
		These four claims put together show that $\alpha^\prime$ belongs to $\LA$, while also distinguishing $M$ and $N$.
		However $|\alpha^\prime|<|\alpha|$, contradicting our initial choice of
		$\alpha$. Thus, no machine $N\in \Im_k$
		can satisfy $M\sim_E N$ and $M\nsim_\LA N$
		at the same time. This completes the proof of our theorem. \qed
	\end{proof}
	
	\subsection{Cores and Covers}
	
	Analogously to classical conformance testing algorithms, our proposed methods rely on the initial construction of ``cover'' of relevant locations. For this we use a notion of 
	core equivalent to the one appearing in 
	\cite{petrenkoTestingPartialDeterministic2005}.
	\par

	We say that a set $V\subseteq \LA$ 
	is \textbf{well-founded} if $\emptyword\in V$.
	Let $V$ be a well-founded set.
	For a word $\alpha\in \LA$, we define $|\alpha|_\SubV$ as the length
	of the shortest suffix $\gamma$ satisfying $\beta\gamma = \alpha$, for
	some $\beta\in V$. Given words $\alpha,\beta$, we write $\beta\leq_\SubV \alpha$ if
	$\beta\leq \alpha$ and additionally
	$\beta < \gamma < \alpha$ holds for no
	sequence $\gamma\in V$. 
	Intuitively, this means that $\beta$
	lies along the shortest path from $V$
	to $\alpha$. It is straightforward to see that
	$\leq_\SubV$ constitutes a partial order
	over $\LA$.
	Finally, we put
	$\sfrac{b}{\beta}\preceq_\SubV \sfrac{a}{\alpha}$ 
	for a pair of nodes if $\beta \leq_\SubV \alpha$, in addition to 
	$\sfrac{b}{\beta}\preceq \sfrac{a}{\alpha}$.
	Given a ranking $R$, we define $R\preceq_\SubV \sfrac{a}{\alpha}$
	analogously as before. 
	\par
	
	We call a set of
	locations $Q\subseteq \Locs$ a \textbf{core}, 
	if for all $(s,a)\in \Locs$
	there is some $(t,b)\in Q$ with
	$(s,a) \sim (t,b)$ and $b\sqsupseteq a$.
	A \textbf{core cover} is a well-founded set $V\subseteq \LA$ for which the set $\{ \, (s,a) \, | \,
	\exists \alpha\in V, \, (s,a) \in \Delta_{S\times A}(\alpha) \, \}$ is a core. \par

	\subsection{Certificates and Main Condition
	for Completeness}
	
	Here we give our main 
	sufficient condition for suite completeness.
	This condition is enforced constructively
	by our proposed algorithms (Section~\ref{sec:algs}), ensuring that
	they produce $k$-complete suites, as required. 
	For the remainder of the section, we fix a 
	core $Q\subseteq \Locs$ and a corresponding
	cover $V\subseteq E$, in addition to 
	$M, A, E, \sqsubseteq$, which were set beforehand. \par

	Given a node ranking $R\subseteq \Gamma$, a \textbf{basis} for $R$ is another set of nodes $B\subseteq \Gamma(V)$ satisfying the following two properties: 
	(1) Nodes in $B$ correspond to pair-wise
	incompatible locations. That is, $(\delta_M(\alpha),a)\nsim (\delta_M(\beta), b)$ for all 
	nodes  
	$\sfrac{a}{\alpha}, \sfrac{b}{\beta}\in B$. 
	(2) Whenever 
	$(\delta_M(\alpha),a)\sim (\delta_M(\beta), b)$ holds for some $\sfrac{a}{\alpha}\in B$, $\sfrac{\beta}{b}\in C$, it follows that $a\sqsupseteq b$. Intuitively, this means
	that $B$ represents more ``testable'' locations than $R$. 
	\par
	
	A \textbf{redundancy certificate} for a node $\sfrac{a}{\alpha}$ is a pair $(R,B)$ where $R\subseteq \Gamma(E\setminus V)$ is a monotonous ranking satisfying $R\preceq_\SubV \sfrac{a}{\alpha}$, and $B\subseteq \Gamma(V)$ is a basis for $R$. 
	Note that according to this definition $R$ and
	$B$ are disjoint. 
	Analogously to rankings, certificate is \textbf{flat} if
	all nodes in $R\cup B$ correspond
	to the same state $a\in S_A$.
	We say that a node $\sfrac{a}{\alpha}\in \Gamma$ is $k$-\textbf{redundant} if there is some
	redundancy certificate $(R,B)$ for the node $\sfrac{a}{\alpha}$ which satisfies $|R| + |B|= k$ and
	$E$ is incompatibility preserving w.r.t. $R\cup B$. Analogously as with $k$-saturated sequences, we say that a test $\alpha\in \LA$ is $k$-redundant if all the nodes $\sfrac{a}{\alpha}$,
	where $a\in \Delta_A(\alpha)$,
	are $k$-redundant themselves. 
	\begin{theorem}
		\label{thm:main_w_cover}
		Suppose that all tests $\alpha\in \LA \setminus E$ have a $(k+1)$-redundant prefix $\beta\in E$, satisfying
		$\beta \leq_\SubV \alpha$. Then $E$
		is $k$-complete. 
	\end{theorem}
	
	The proof is similar to the one of 
	Theorem~\ref{thm:main_no_cover}.
	The main argument relies on showing that distinguishing sequences $\alpha$ outside of $E$ can be ``shrunk'' as well. The two main differences are that
	now the relevant measure of size
	is $|\alpha|_\SubV$ rather than $|\alpha|$, and that in the combinatorial arguments we exploit the
	sizes of certificates $(R,B)$, rather than those of rankings $R$, as before. 
	The full proof can be found at
	Appendix~\ref{ap:thm_main}.

	\section{Proposed Algorithms} \label{sec:algs}
	In this section we give high-level descriptions of two algorithms for the restricted conformance problem. 
	Let $M$ be an specification machine and $A$
	a context automaton, as before. We present two algorithms
	for the restricted conformance testing
	problem, dubbed
	\textsc{Simple} and 
	\textsc{Complex}, which use the theory
	developed so far. 
	Both procedures mainly differ in
	whether they attempt to exploit the language
	inclusion relation  over $S_A$. \par


	\subsection{Simple Variant}
	\label{sec:simple}
    Our procedure \textsc{Simple} uses a generalization of the concept of harmonized identifiers adapted to our context.
	A family of \textbf{harmonized identifiers} is given by 
	a set of words $W_{(s,a)}$ for each location $(s,a)\in S_{M\times A}$ satisfying (1) $W_{(s,a)}\subseteq \LA(a)$, (2)
	whenever $(s,a)\nsim (t,a)$ for some $(s,a),(t,a)\in S_{M\times A}$, some $\alpha\in W_{(s,a)}\cap W_{(t,a)}$ witnesses
	$(s,a)\nsim (t,a)$. Note that the sets $W_{(s,a)}$
	only need to distinguish $(s,a)$ from other locations 
	corresponding to the same context state $a$. 
	\par
	Algorithm~\ref{alg:first} shows the basic structure of \textsc{Simple}.
	The algorithm constructs a $k$-complete suite $E$ by successively adding various sequences to it. We assume $E$ to be prefix-closed throughout the exposition. Hence, whenever we include a test $\alpha$ in $E$, all its prefixes are implicitly added as well. We initialize the suite $E$ 
	to a cover $V$ of some core $Q$ (line~\ref{lin:basic_init_suite}). The routine 
	$\textsc{WeakCore}()$ simply selects one
	location $(s,a)$ from each equivalence class
	$S_{M\times A}/ \simeqq$, and $\textsc{Cover}(Q)$ explores $\LA$ in a breath-first fashion until all locations 
	in $Q$ have been visited. Afterwards, we 
	compute a family of harmonized identifiers $W_{(s,a)}$, and enlarge $E$ by appending
	them to suitable sequences $\alpha \in V$
	(line~\ref{lin:first_bg_dist_cover}).
	Finally we expand $E$ in a depth-first way
	starting from each word $\alpha_\SubV\in V$ (line~\ref{lin:fst_bg_expl}).
	\begin{algorithm}
		\caption{\textsc{Simple}($M,A,k$)}\label{alg:first}
		\hspace*{\algorithmicindent} \textbf{Input}
		A specification machine $M$, context automaton $A$, 
		and a bound $k$. \\
		\hspace*{\algorithmicindent} \textbf{Output} A $k$-complete
		suite $E$ for $M$ in the context of $A$.
		\begin{algorithmic}[1]
			\State $Q\gets \textsc{WeakCore}()$
			\State $V, \, toCvr \gets \textsc{Cover}(Q)$
			\Comment{
				\parbox[t]{.4\linewidth}{
					$toCvr$ is a map $Q \rightarrow V$ where $(s,a)\in \Delta_{M\times A}(toCvr(s,a))$
				}
			} 
			\State $E\gets V$ \label{lin:basic_init_suite}
			\State $ \{W_{(s,a)}\}_{(s,a)} \gets$ family of harmonized identifiers \label{ln:bs_harm}
			\State \textbf{for all } $(s,a) \in Q$ \textbf{ do }
			$E\gets E\cup \alpha\,W_{(s,a)}$, where $\alpha\coloneqq toCvr(s,a)$ 
			\label{lin:first_bg_dist_cover}{\tiny {\tiny }}			
			\State \textbf{for all} $\alpha \in V$ \textbf{do}
			$\alpha_\SubV\gets \alpha$, and $\textsc{Explore}(\emptyword)$ \label{lin:fst_bg_expl}
			\State \Return $E$
		\end{algorithmic}
	\end{algorithm}

	The final depth-first exploration
	carried out in the routine \textsc{Explore}$(\beta)$, shown in Algorithm~\ref{alg:explore1}. The search
	conducted in a recursive manner starting from $\alpha_\SubV$. This is
	done by expanding a candidate suffix $\beta$ successively. For this purpose, we examine each possible continuation $\alpha_\SubV\beta i$ 
	and determine whether the search space can be pruned at that point. We decide to stop exploring from
	$\alpha_\SubV\beta i$ if the sequence can be made $(k+1)$-redundant by adding suitable distinguishing sequences. This is done a big enough redundancy certificate for each node $\sfrac{a}{\alpha_\SubV\beta}\in \Gamma$ via
	$\textsc{SearcCerts}(\beta)$, and
	making $E$ incompatibility preserving w.r.t. 
	these certificates in $\textsc{ExploitCert}(R,B)$. We give a more detailed view of those steps. \par
	\begin{algorithm}
		\caption{\textsc{Explore}($\beta$)}
		\label{alg:explore1}
		\hspace*{\algorithmicindent} \textbf{Input} a suffix $\beta$ with
		$\alpha_\SubV\beta\in \LA$.
		\begin{algorithmic}[1]
			\ForAll{ inputs $i\in I_M$ with $\alpha_\SubV\beta i \in \LA \setminus V$}
			
			\State $Certs \gets \textsc{SearcCerts}(\beta i)$.
			\If{$Certs \neq false$}
			\State add $\alpha_\SubV \beta i$ to $E$
			\State \textbf{for all }$(R,B)\in Certs$ \textbf{do}	
			\textsc{ExploitCert}$(R,B)$
			\Else \,\,  \textsc{Explore}$(\beta i)$

			\EndIf
			\EndFor
			\State \Return
		\end{algorithmic}
	\end{algorithm}
	
	The function \textsc{SearcCerts}$(
	\beta)$, shown in Algorithm~\ref{alg:search_certs}, 
	attempts to find a redundancy certificate $(R_a,B_a)$ satisfying $|R_a|+|B_a| = k + 1$ for each 
	node $\sfrac{a}{\alpha_\SubV\beta}\in \Gamma$.
	If it succeeds, the family of certificates $(R_a, B_a)$ is returned.
	Otherwise, it just returns $false$.
	The search of a certificate
	$(R_a,B_a)$ for a node $\sfrac{a}{\alpha_\SubV\beta}$ is divided in two stages. 
	First, a set $Rankings$ of candidate rankings satisfying $R\preceq_\SubV 
	\sfrac{a}{\alpha_\SubV\beta}$ is constructed
	via $\textsc{BuildRankings}(\beta,a)$. Afterwards, for each
	ranking $R\in Rankings$ we find a suitable
	basis using the routine $\textsc{Basis}(R)$, and we check whether  
	$|R|+|\textsc{Basis}(R)|\geq k+1$. \par

	\begin{algorithm}[tbh]
		\caption{\textsc{SearchCerts}$(\beta)$} \label{alg:search_certs}
		\textbf{Input} a suffix $\beta$ with
		$\alpha_\SubV\beta\in \LA$.\\
		\textbf{Output} a set $Certs$ of redundancy certificates for $\alpha_\SubV\beta$, or $false$
		\begin{algorithmic}[1]
			\State $Certs \gets \{ \}$
			\ForAll{$a\in \Delta_A(\alpha_\SubV\beta)$}
			\State $Rankings \gets \textsc{BuildRankings}(\beta,a)$
			\If{there for some $R\in Rankings$ with $|R|+|\textsc{Basis}(R)| \geq k+1$} \label{lin:certs_if}
			\State add $(R,\textsc{Basis}(R))$ to $Certs$
			\Else \,\, \Return $false$ 
			\EndIf
			\EndFor
			\State \Return $Certs$.
		\end{algorithmic}
	\end{algorithm}
	
	In this variant, $\textsc{BuildRankings}(\beta,a)$
	builds a family of flat rankings through a linear 
	scanning of the nodes
	$\sfrac{c}{\varphi}\preceq_\SubV \sfrac{a}{\alpha_\SubV\beta}$. 
	Given a flat ranking $R\coloneqq
	(\sfrac{c}{\varphi_i})_{i=1}^\ell$, for a fixed $c\in S_A$, the method $\textsc{Basis}(R)$ constructs a basis
	for $R$ simply by finding all
	locations of the form $(s,c)$ in the core $Q$. Finally, 
	the function $\textsc{ExploitCert}(R,B)$
	is tasked with making $E$ incompatibility preserving w.r.t. a given
	flat certificate $(R,B)$ by adding several distinguishing
	sequences to $E$. 
	
	\begin{algorithm}[tbh]
		\caption{\textsc{BuildRankings},
		\textsc{Basis}, \textsc{ExploitCert}
		 \hfill (\textsc{Simple}'s version)} \label{algs:basic_version}
		\begin{algorithmic}[1]
			\Procedure{BuildRankings}{$\beta,a$}
			\Statex \textbf{Input} a suffix $\beta$ with
			$\alpha_\SubV\beta\in \LA$, and a state
			$a\in \Delta_A(\alpha_\SubV\beta)$ 
			\Statex \textbf{Output} A set $Rankings$ of
			constant rankings $R\preceq_\SubV \frac{a}{\alpha_\SubV\beta}$.
			\State $Rankings \gets \{ \}$
			\State initialize empty rankings
			$R_{b_1}, R_{b_2}, \dots$ for all 
			$b_i\in S_A$.
			\State $\Omega \gets $ set of nodes
			$\sfrac{c}{\varphi} \preceq_\SubV \sfrac{a}{\alpha_\SubV\beta}$.
			\ForAll{ $j = 1,2, \dots, |\beta|$, and \textbf{all} $(\alpha_\SubV\cdot\beta_{\leq j},b)\in 
				\Omega$}
			\State append $\sfrac{b}{\alpha_\SubV\cdot\beta_{\leq j}}$, to $R_b$.
			\EndFor
			\State $Rankings\gets \{R_{b}\}_{b\in S_A}$
			\EndProcedure
			\Statex
			\Procedure{Basis}{$R$}
			\Statex	\textbf{Input} A flat ranking $R=(\sfrac{c}{\varphi_j})_{j=1}^\ell \subseteq \Gamma$, for some $c\in S_A$.
			\Statex \textbf{Output} A basis $B$ for $R$.
			\State $B \gets \{ \}$
			\State \textbf{for all }$(s,c) \in Q$, add $(toCvr(s,c),c)$ to $B$
			\State \Return $B$
			\EndProcedure
			\Statex
	
			\Procedure{ExploitCert}{$R,B$}
			\Statex A flat redundancy certificate $(R,B)$.
			\ForAll{$\sfrac{c}{\varphi} \in C$}
			\State $s\gets \delta_M(\varphi)$
			\State add $\varphi W(s,c)$ to $E$.
			\EndFor
			\EndProcedure			
		\end{algorithmic}
	\end{algorithm}

	\subsection{Complex Variant}
	\label{sec:complex}
	The basic structure of the method \textsc{Complex} is 
	is largely 
	similar that of \textsc{Simple}. The
	main difference is that \textsc{Complex} 
    takes an additional parameter $\sqsubseteq$,
    which is an under under-approximation
    of language inclusion over $S_A$.
    The goal of \textsc{Complex}
    is to
    exploit $\sqsubseteq$ to 
    obtain a possibly more reduced
	suite than \textsc{Simple}.
	The detailed description of the algorithm is mostly technical
	in nature an can be found in Appendix~\ref{ap:complex}. 
	\textsc{Advanced} uses $\sqsubseteq$ 
	two main different ways: 
	(1) It uses $\sqsubseteq$ for
	computing the core $Q$, yielding
	a possibly smaller initial cover than \textsc{Simple}. (2) It uses $\sqsubseteq$ to search for non-flat chains and certificates.
	This potentially allows
	\textsc{Advanced} to prune the exploration process space earlier than \textsc{Simple}. \par
	The procedure however, shows two
	main disadvantages with respect
	to the simpler variant. The first
	is that searching for general certificates costs more
	time than searching just for
	flat ones, as \textsc{Simple} does.
	The second is that making a suite
	incompatibility-preserving w.r.t.
	general certificates requires more
	involved strategies for adding distinguishing suffixes. Here the idea of using harmonized identifiers does not work,
	as one needs to distinguish locations
	$(s,a),(t,b)$ for $a\neq b$,
	and \textsc{Complex} potentially adds more distinguishing sequences, or longer ones. 
	
	\subsection{Complexity Bounds} \label{sec:bounds}
	In this section we study the complexity of our procedure
	\textsc{Simple} both in terms of
	time and sizes of the output suites.
	We also	give the related expressions for \textsc{Complex}.
	Two notable aspects come out from of this
	analysis. One is that our methods
	avoid the addition of exponential-length
	tests, issue which the
	State-Counting approach
	\cite{petrenkoTestingPartialDeterministic2005} suffered from. The second is that our
	proposed techniques
	spend polynomial
	time in the generation of each test
	sequence, unlike the SAT-based approach
	from \cite{petrenkoLearningCommunicatingState2019a}.
	\par
	Fix $M,A,k$, with $k\geq |S_M|$.
	First we sketch 
	a bound for the total number of tests in the suite
	\textsc{Simple}$(M,A,k)$. 
	Let $n_{(M\times A)} \coloneqq|\sfrac{S_{M\times A}}{\cong}|$.
	The core $Q$ contains a location
	$(s,a)$ for each class in  $\sfrac{S_{M\times A}}{\cong}$. Thus, $|Q|\leq 
	n_{M\times A}\leq |S_M||S_A|$, and a cover $V$ for $Q$ contains at most
	$|S_M||S_A|$ words.
    Now we give a 
    bound the depth of the exploration process carried out in \textsc{Explore}. 
	The following result refers 
	to the scope of \textsc{Simple}. Its proof
	can be found at Appendix~\ref{ap:length}
	
	\begin{theorem} 
	\label{thm:length}
	Fix $\alpha_V\in V$. Let $\beta$ be a suffix
	with $\alpha_V\beta\in \LA$ and $|\beta|=k|S_A|-
	n_{M\times A} + 1$
	Then the method \textsc{SearchCerts}$(\beta)$ does not
	return $false$.
	\end{theorem}
	
	Let $e\coloneqq k|S_A| - n_{M\times A}$. The
	parameter $e$ plays a similar role
	in this analysis
	to the number of extra states in traditional
	conformance testing. 
	Last result shows that the $\textsc{Explore}$ in the worst case may
 	add possible suffixes $\beta$ of size $e + 1$ to each word $\alpha_\SubV \in V$.
 	This yields potentially $|S_M||S_A||I_M|^{e+1}$ sequences
 	of the form $\alpha_\SubV\beta$.
 	For each of these, \textsc{Simple} appends
 	appends potentially $|S_A|$ identifiers
 	$W_{(s,a)}$, either during its initial phase
 	or during $\textsc{ExploitCerts}$. 
 	This yields an upper bound of
    $|S_A|^2|S_M|^2 |S_M| |I_M|^{e + 1}|$
    tests in the suite returned by \textsc{Simple}.
    \par
    
    To obtain the total number of symbols produced
    by $\textsc{Simple}$ we multiply last 
    bound by the maximum size of a test
    in the suite. Without loss of generality,
    tests generated in $\textsc{Simple}$
    are of the form
    $\alpha_\SubV \beta \gamma$, where 
    $\alpha_\SubV$ belongs to the cover $V$,
    $\beta$ is an arbitrary suffix with $|\beta| \leq e + 1$, and $\gamma$ is a distinguishing sequence belonging
    to some haromonized identifier
    $W_{(s,a)}$. Clearly, $|\alpha_\SubV|\leq
    |S_A||S_M|$, and using
    Theorem~\ref{thm:seq} yields
    $\gamma\leq |S_A||S_M|$ as well.
    Putting everything together we get $|\alpha_\SubV \beta \gamma|
    \leq 3|S_A|k$. This gives us a bound expression of $O(k|S_A|^3|S_M|^2 |I_M|^{e + 1}|$ symbols generated in \textsc{Simple}. \par
    We note that the bounds obtained 
    for \textsc{Simple} are optimal, in the 
    sense that whenever $A$ is the universal
    NFA with one state, we recover the bounds
    for the W-method, discussed in Section~\ref{sec:classic_testing}. The
    time cost of analysis of \textsc{Simple}
    can be gotten from examining the 
    routines \textsc{SearchCerts} and
	\textsc{ExploitCert}. This can be seen
	in more detail in Appendix~\ref{ap:cost}.
	The resulting time cost is
	$O((k|S_A|^3|S_M|^3 + |S_A|^4|S_M|)e|I_M|^{e+1})$. \par
	For completeness sake we briefly discuss
	the complexity analysis of \textsc{Complex}.
	The bounds for number of tests and symbols
	obtained for \textsc{Simple} also apply
	for this second variant following similar
	arguments. The time-cost of the procedure
	is covered in Appendix~\ref{ap:cost},
	and is given by $O((k|S_A|^3|S_M|^3 + |S_A|^5|S_M|)e|I_M|^{e+1})$.

	\par

	\section{Experimental Results}
	
	Our proposed methods
	were motivated by the task of testing a 
	component with observable interfaces and non-controllable inputs. 
	During our experiments, we evaluated our 
	techniques on the problem of testing
	the tail $T$ of a cascade composition $T\circ H$ (Problem~\ref{prob:tail}).
	For this, we use the reduction
	described in Section~\ref{sec:tail}, which
	transforms the head $H$
	into a suitable NFA $A$. 
	We aim to answer the following questions:
	(1) How do our techniques compare against the 
	baseline method presented in
	Section~\ref{sec:tail}?
	(2) How do the sizes of the component machines and the number of extra states influence
	our methods? Finally, the theory
	developed in 
	Section~\ref{sec:theory}
	allows for a natural application
	of approximate techniques for
	NFA reduction and
	language-inclusion. Hence, our last 
	question is: (3) what kind of
	impact do those strategies have? 
	We describe now our experimental setup.
	Our benchmarks consist of randomly constructed cascades
	of Mealy machines, formed by a head $H$, and a tail $T$, 
	where $O_H=I_T$.
	We say a cascade is of size $n \times m$
	if $|S_H|=n$ and $|S_T|=m$.
    To construct
    the random benchmarks, we utilized
    the generator in FSMLib \cite{souchaFSMLib}, 
	which produces reduced connected Mealy machines with given alphabet sizes and number of states. 
	All experiments were run on an Intel Core i5-6200U (2.30GHz) machine with a limit of 4GB RAM memory,	and a time limit of $3$ minutes
	 \par

    In order to answer the first question, we
    implemented \textsc{Simple}  (Section~\ref{sec:simple}) and compared
    it against the testing of the composite
    machine described in Section~\ref{sec:tail}.
    To represent this baseline, 
    we applied the H-method \cite{dorofeevaImprovedConformanceTesting2005}
    on the composite machine $P$, using the implementation provided by FSMLib. \par
    
    	\begin{figure}[ht]
		\centering
		\subfloat[ Number of symbols.
		\label{fig:small_sym}]{
			\includegraphics[width=0.45\textwidth, valign=b]{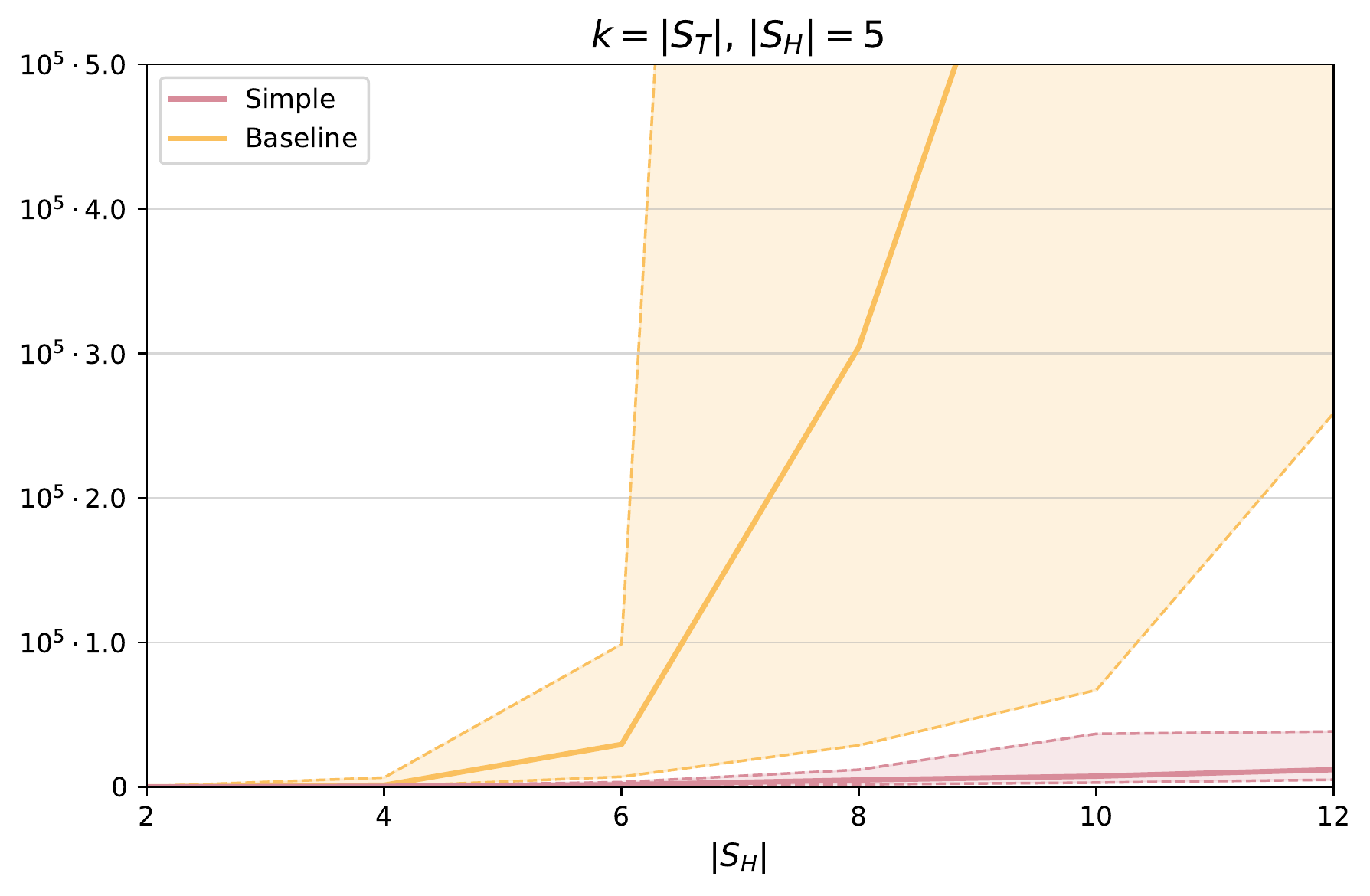}
		}
		\hfill
		\subfloat[Execution time (s).\label{fig:small_times}]{
			\includegraphics[width=0.45\textwidth,
			valign=b]{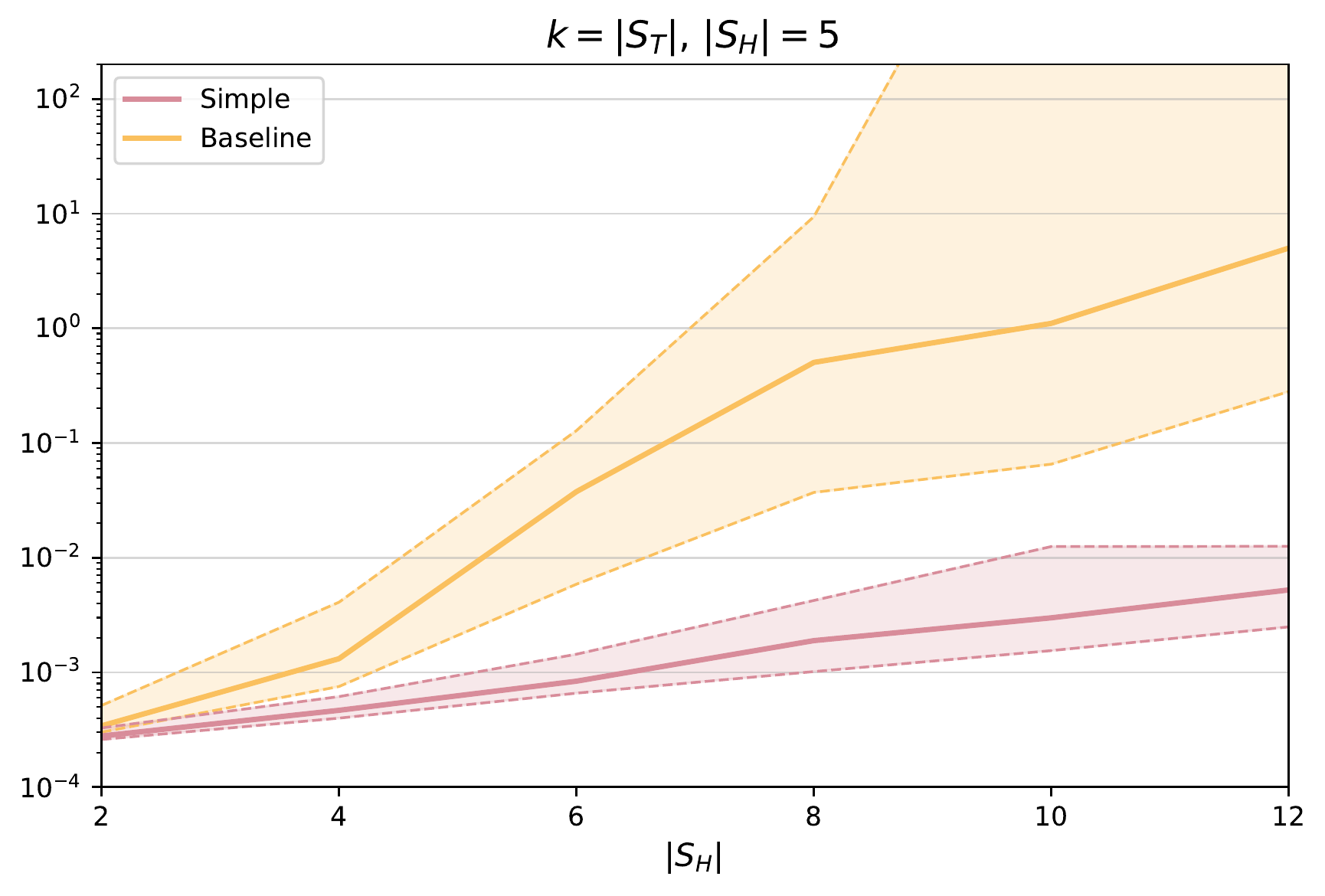}
		}
		\caption{Comparison between \textsc{Simple} and
		the testing of the composite machine.}
		\label{fig:short}
		
	\end{figure}
    
    In Figure~\ref{fig:short}, we compare total numbers of symbols and execution times
	for \textsc{Simple} and the baseline method. For each tail size $|S_T|=2,4,6,8,10,12$, we generated one hundred cascades where 
	$|S_H|=5$ and all alphabets were of
	size $4$. We considered no extra states
	in these experiments. That is, we aimed
	for $k$-complete suites for $T$, where 
	$k=|S_T|$. Solid lines in our graphs represent
	median quantities, and areas
	around those lines are enclosed
	by the 25-th and 75-th percentiles of their respective metrics.  
	We conclude that our proposed method, \textsc{Simple},
	greatly outperforms the testing of the composite machine	in both selected criteria. The main problem the baseline method
	encountered was the space limitation. Already with cascades
	of size $5\times 10$, $31\%$ of the experiments ran out of memory. The root cause of this was the blow-up of extra states, discussed in 
	Section~\ref{sec:tail}.	Among the
	$5\times 12$ benchmarks, the amount
	of extra states considered by the baseline
	was bigger than  $9$ a $27\%$ of the times.\par

	In order
	to study the potential benefits
	of NFA reduction and language-inclusion techniques, we 
    implemented an additional
    algorithm representing our best attempt at the
	gray-box testing problem. 
	Here, first 
	we optimize $A$'s with the approximate method
	implemented in the tool
	\textit{Reduce} \cite{clementeEfficientReductionNondeterministic2019a}. 
	Afterwards, we compute the so-called
	``look-ahead forward direct simulation relation'', introduced
	in \cite{clementeEfficientReductionNondeterministic2019a}, which gives us an under-approximation $\sqsubseteq$ of language inclusion over $S_A$. If this results
	in a trivial relation, we fall back to \textsc{Simple}. Otherwise, we try to exploit $\sqsubseteq$ by calling 
	\textsc{Complex} (Section~\ref{sec:complex}). 
	We dub this whole procedure \textsc{Advanced}.
	The amount of look-ahead used in both the tool \textit{Reduce} and the computation of $\sqsubseteq$ was set to $16$.
	 \par

	\begin{figure}[ht]
	\centering
	\subfloat[ Number of symbols.]{
		\includegraphics[width=0.45\textwidth,valign=b]{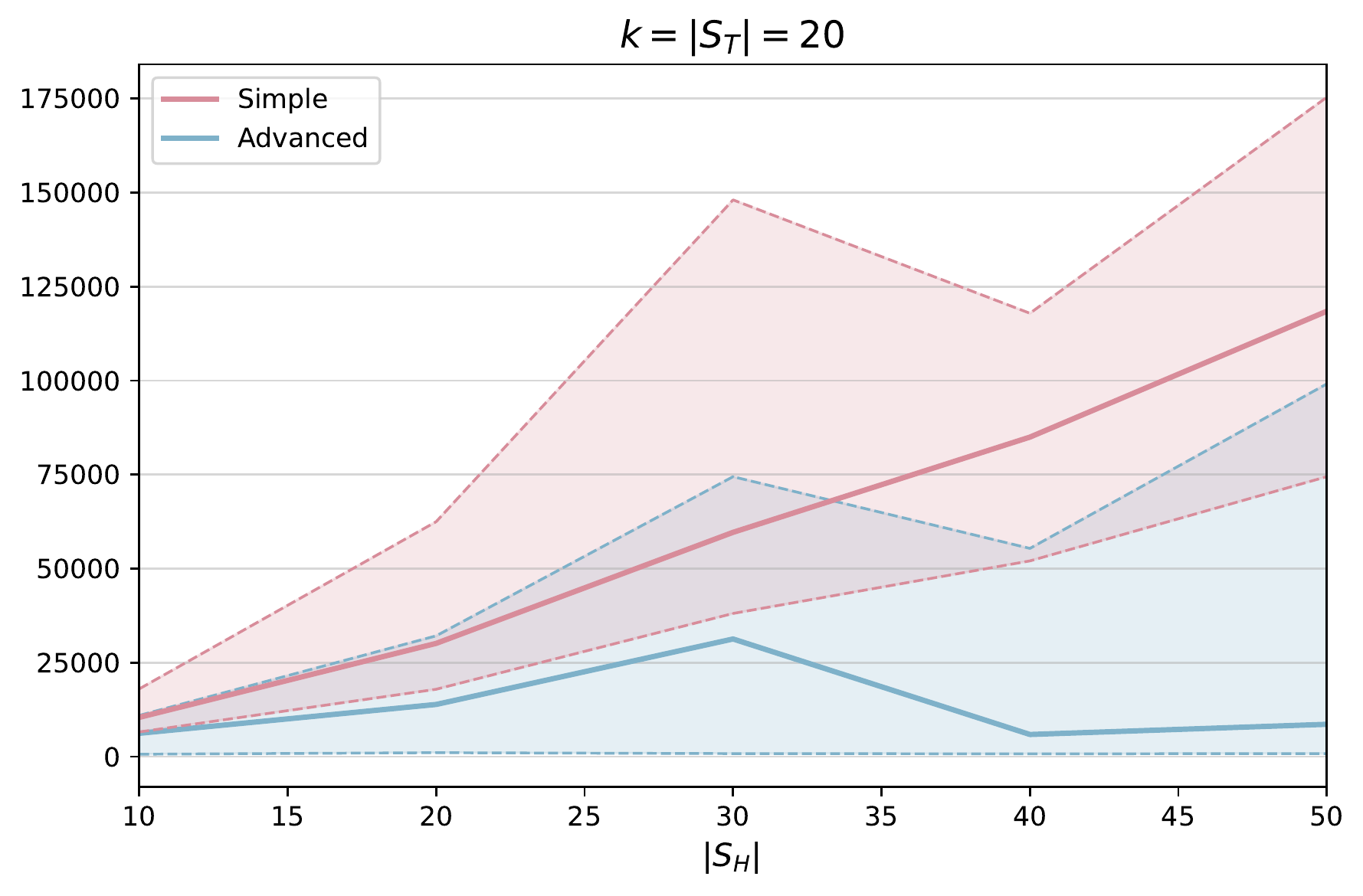}
	}
	\hfill
	\subfloat[Execution time (s)]{
	\includegraphics[width=0.45\textwidth,valign=b]{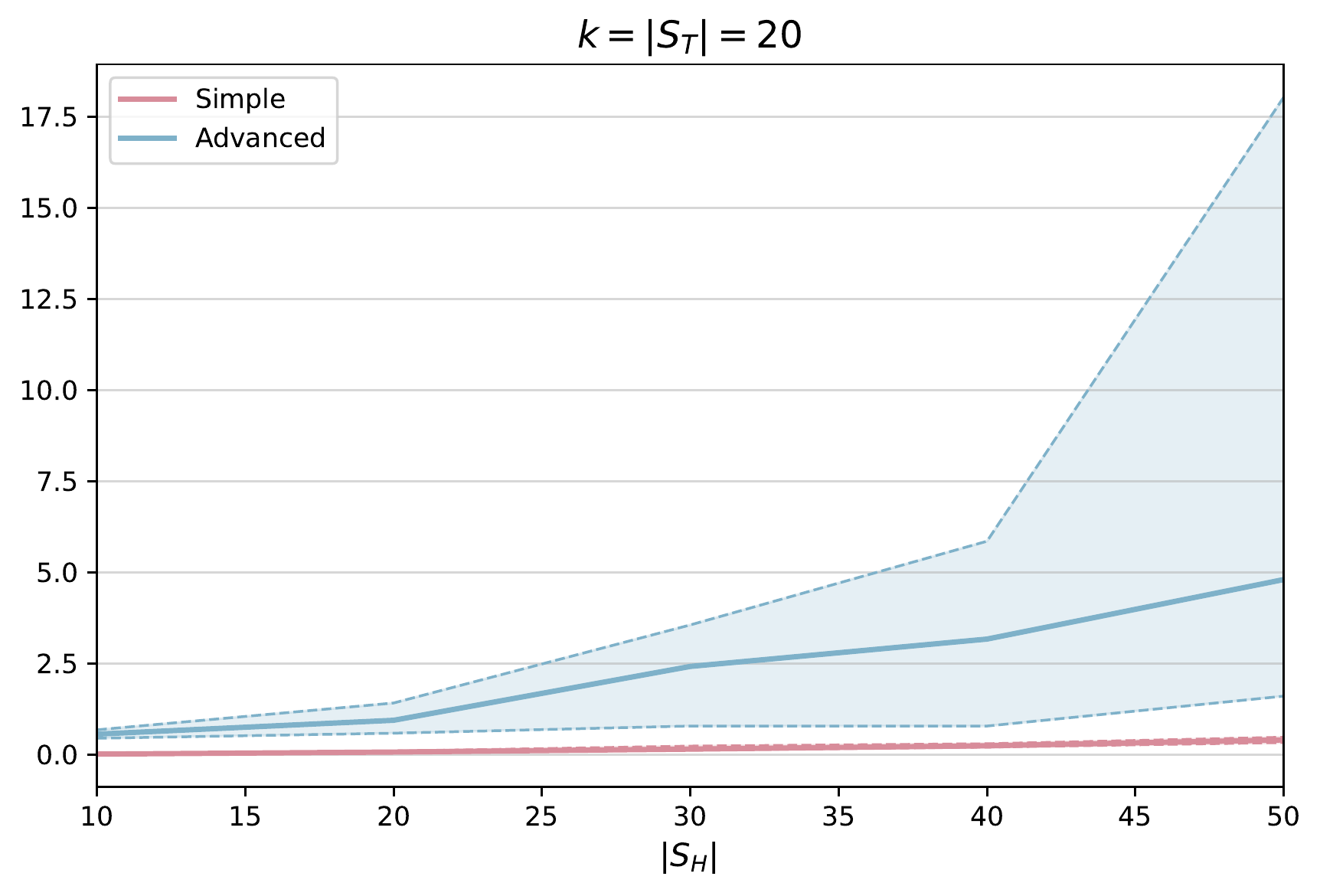}
}
\\
	\centering
	\subfloat[ Number of symbols.]{
		\includegraphics[width=0.45\textwidth,
		valign=b]{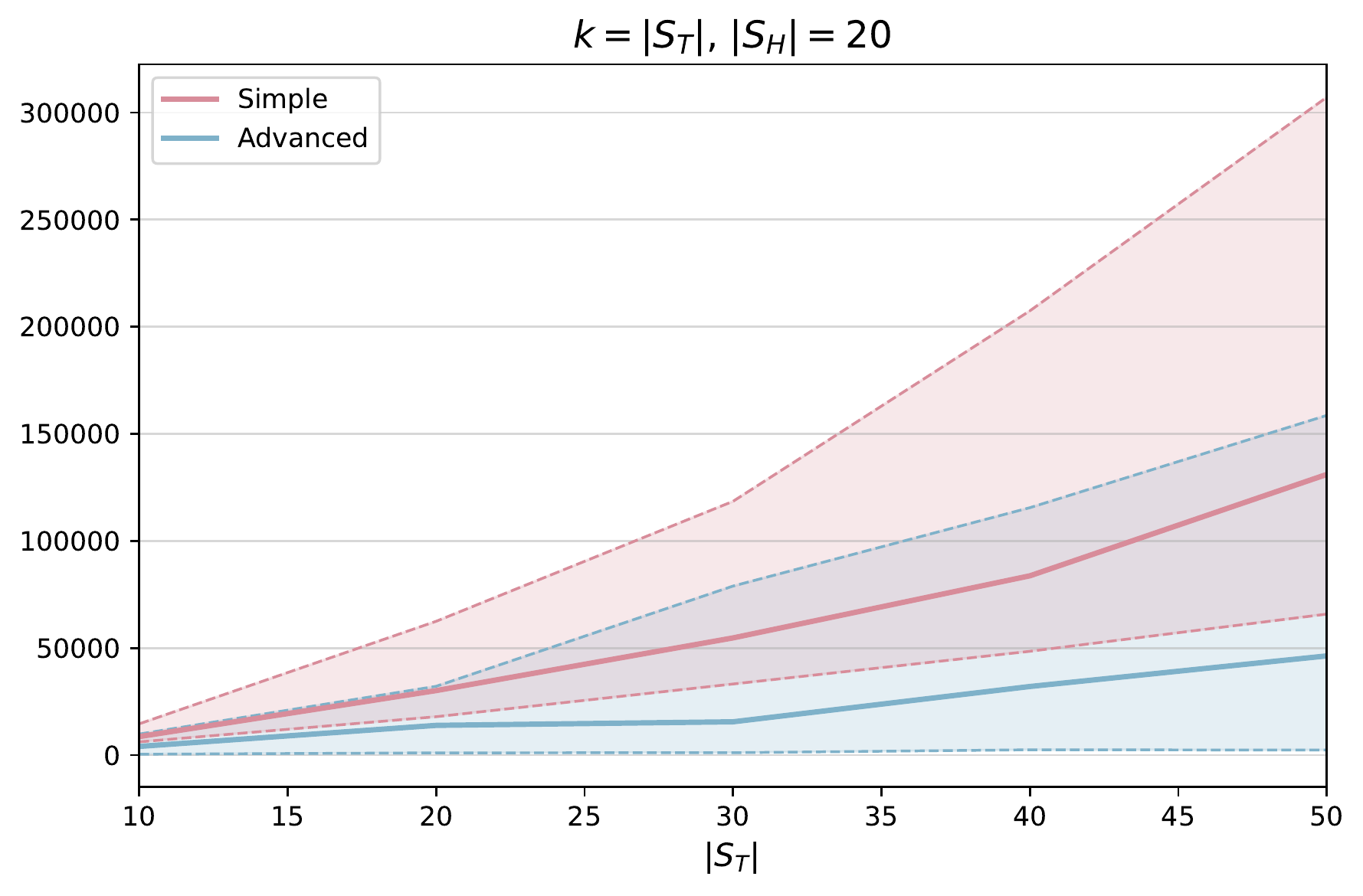}
	}
	\hfill
	\subfloat[Execution time (s)]{
		\includegraphics[width=0.45\textwidth,
		valign=b]{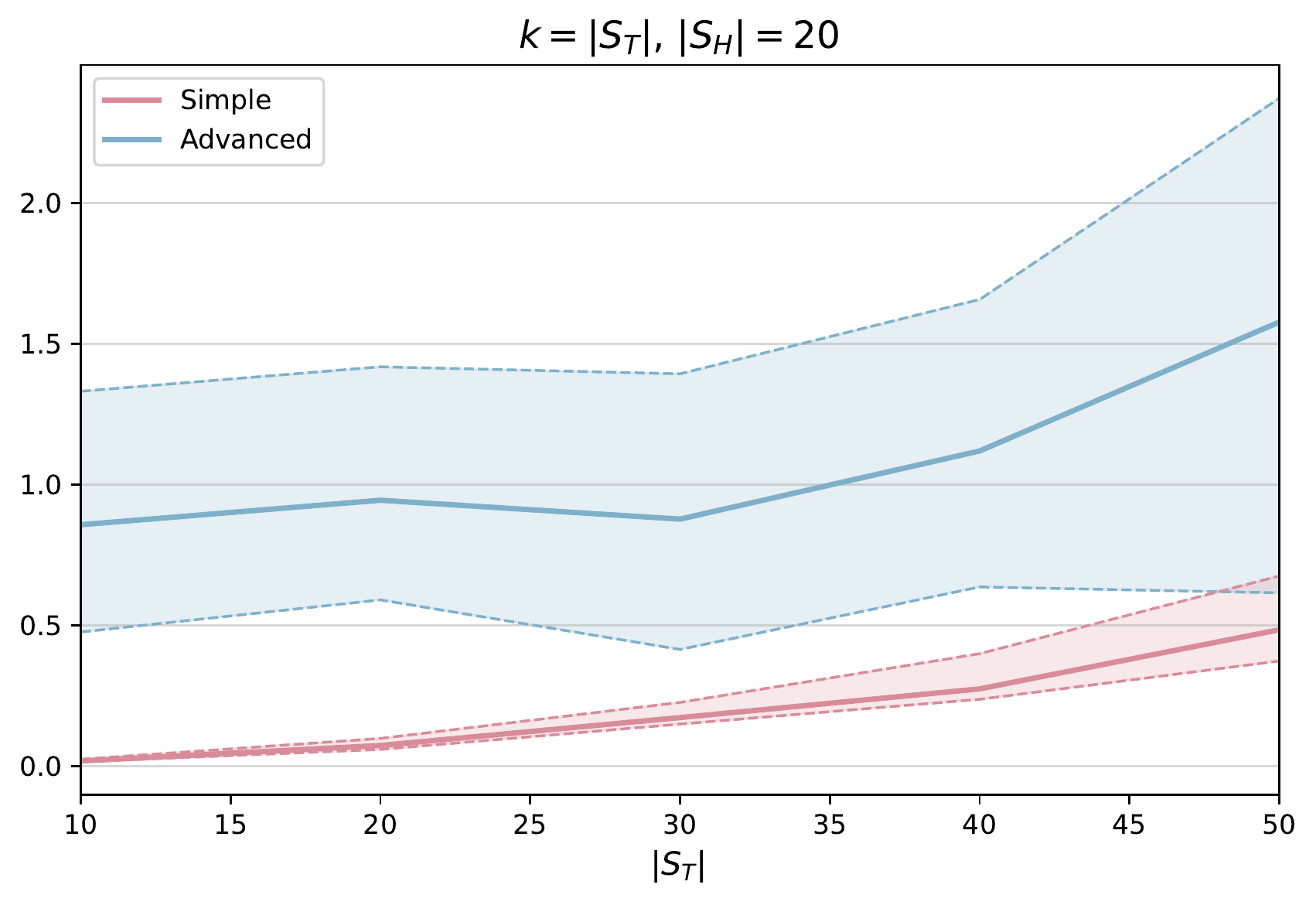}
	}
	\\
	\caption{Performance of our proposed algorithms
	with respect to the size of the head (upper row)
	and the tail (lower row).}
	\label{fig:main_batch}
	
\end{figure}
	To address the rest of our questions, we generated two additional batches of $500$ cascade compositions each.
	In the first, we fixed $|S_H|=20$,
	and generated $100$ benchmarks for each value
	$|S_T|=10,20,30,40,50$.
	In the second followed the same process
	with the roles of $|S_H|$ and $|S_T|$ reversed.
	In order to obtain automata $A$
	where minimization and language-inclusion
	techniques show interesting behaviour, we fixed 
	$|I_H|=6$, and $|I_T|=|O_T|=3$. Experimentally, for values $\sfrac{|I_H|}{|I_T|}$ smaller than
	two, we found those techniques to have no effect on $A$ in
	the majority of times, while for larger values $A$ is easily
	found to be universal. This is 
	consistent
	with the results in \cite{clementeEfficientReductionNondeterministic2019a}.\par
	
	Figure~\ref{fig:main_batch} displays the
	experimental data of \textsc{Simple} and 
	\textsc{Advanced} on this second set of benchmarks, with zero additional states under
	consideration. 
	We do not include the baseline here,
	as it yielded out of memory errors already in
	$80\%$ of $20\times 10$ and $10\times 20$ compositions. The general trend is that \textsc{Advanced} produces much smaller suites than \textsc{Simple} at the cost of a greater execution time. Both aspects of this 
    comparison are more pronounced when $H$ grows than when $T$ does so. We attribute these differences largely to the automata reduction step in \textsc{Advanced}. In 
    $90\%$ of the experiments, the minimization call was responsible $73.4\%$ of \textsc{Advanced}'s execution time, while on half the experiments this number ascends to $95.5\%$. It is worth pointing out that despite producing larger suites, \textsc{Simple} was able to complete the vast majority of the experiments ($927/1000$) in under a second. \par

\begin{figure}[ht]
	\centering
	\subfloat[ Number of symbols.]{
		\includegraphics[width=0.45\textwidth,valign=b]{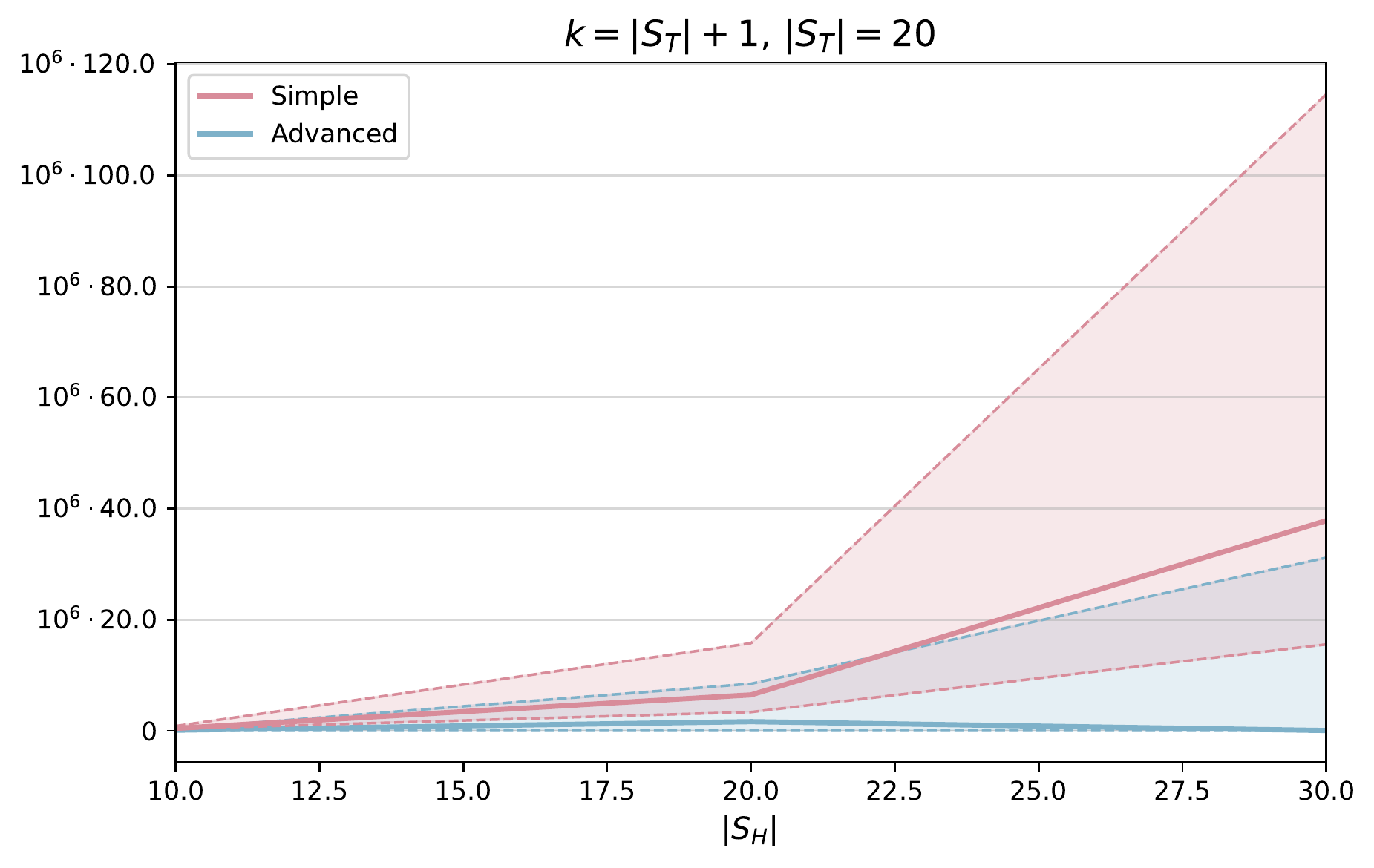}
	}
	\hfill
	\subfloat[Execution time (s)]{
	\includegraphics[width=0.45\textwidth,valign=b]{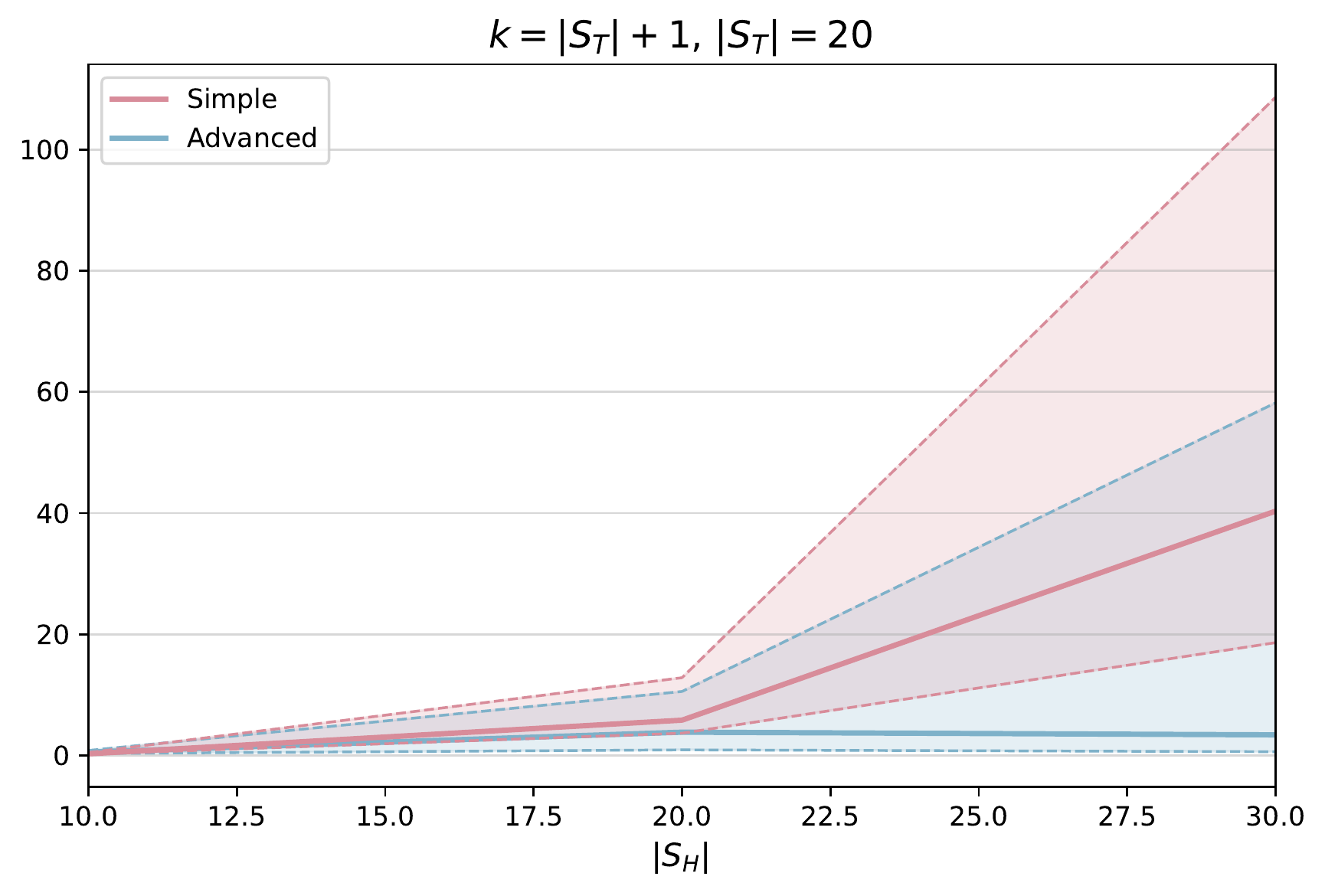}
}
\\
	\centering
	\subfloat[ Number of symbols.]{
		\includegraphics[width=0.45\textwidth,
		valign=b]{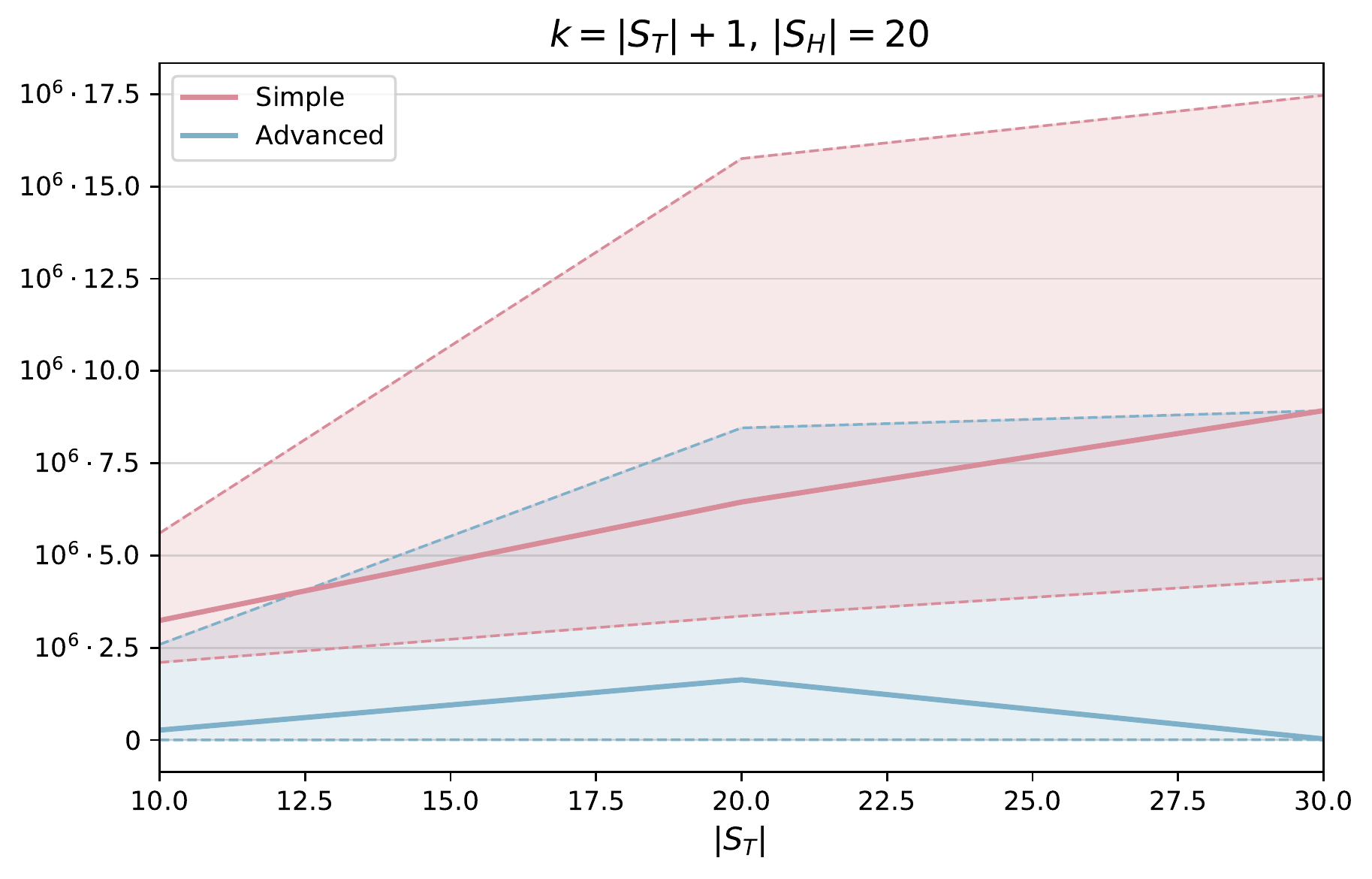}
	}
	\hfill
	\subfloat[Execution time (s)]{
		\includegraphics[width=0.45\textwidth,
		valign=b]{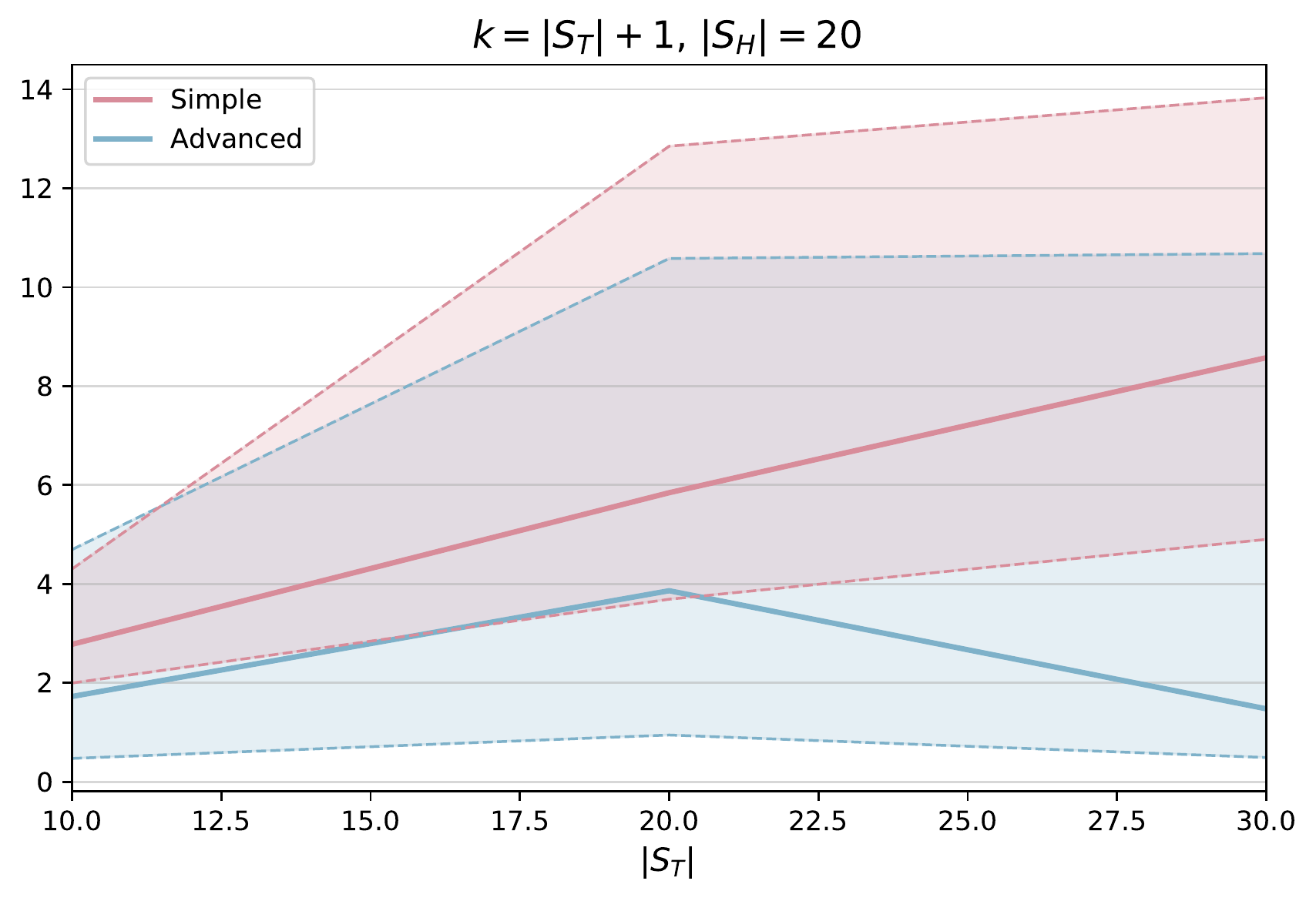}
	}
	\\
	\caption{Performance of our algorithms
	in the presence of extra states. }
	\label{fig:extra}
	
\end{figure}

For Figure~\ref{fig:extra}
we ran again a subset of the previous experiments, but considering one addditional extra state. Out of the original $1000$, we picked the $600$ cascades where head and tail had at most $30$ states.
Here \textsc{Advanced} 
outperforms \textsc{Simple} in both time and number of symbols. Moreover, minimization time
still accounted for a $67 \%$ of \textsc{Advanced}'s execution time in half of the occasions.  
In this case, the initial automata minimization
step seems clearly beneficial. The observed effect of the additional state is drastic both with respect
to execution times and suite sizes. Nevertheless, this impact is much smaller than what our worst-case analyses predict (Section~\ref{sec:bounds}). 
According to those, an additional state could worsen the metrics of both procedures by a factor of $|I_T|^{|S_H|}$. This ascends to around $35\cdot 10^8$ for $|I_T|=3$ and $|S_H|=20$. 
We note that this blowup is
unavoidable for black-box testing
techniques. 
However, the relative increase
between Figure~\ref{fig:main_batch}
and Figure~\ref{fig:extra} 
is not nearly as large.\par

Lastly, to evaluate the effect of the language inclusion relation on our algorithms, we implemented
an additional procedure \textsc{Simple+Reduce},
which just calls \textsc{Simple} after the initial
NFA reduction. We ran the experiments
of Figure~\ref{fig:main_batch}
and Figure~\ref{fig:extra} on this method,
and compared it against \textsc{Advanced}. 
We note that in about $60\%$ of the experiments
both methods performed the same operations, 
as the relation $\sqsubseteq$ obtained from $A$
was trivial. For the remaining $40\%$ of the cases, we computed
the ratio of symbols produced by \textsc{Advanced}
to symbols produced by \textsc{Simple+Reduce}. 
This information is summarized in Table~\ref{tab:first}.
We observe that in $75\%$ of the times
exploiting $\sqsubseteq$ by means of \textsc{Advanced} was either noticeably
beneficial or had almost no effects. However,
in about $10\%$ of the cases the impact was clearly 
negative. 
\begin{table}

\begin{center} 
	\begin{tabular}{|| c || c | c | c | c | c | c | c| c||} 
		\hline
		\multicolumn{9}{||c||}{\textsc{Advanced} / \textsc{Simple}
		+ \textit{Reduce}} \\ [0.5ex] 
		\hline				
		Extra States & 
		$1\%$  & $10\%$ & $25\%$ & $50\%$ & $75\%$ & $90\%$ & $99\%$ & $100\%$\\
		\hline\hline
		$k=|S_M|$ & ~0.047~ & ~0.377~ &~ 0.761~ & ~1.0~~~
		 & ~1.045~ & ~1.301~ & ~1.971~ & ~2.973~ \\ 
		\hline
		$k=|S_M|+1$ & 
		0.006  & 0.249 & 0.515 & 1.0 & 1.009 & 1.360 & 1.926 & $\infty$\\
		
		\hline
		
		\hline
	\end{tabular}
\end{center}
\caption{\label{tab:first} Symbols produced by \textsc{Advanced}
over symbols produced by \textsc{Simple+Reduce}}
\end{table}

	
%
%

\newpage
\bibliographystyle{splncs04}
\bibliography{main}

\begin{thebibliography}{10}
\providecommand{\url}[1]{\texttt{#1}}
\providecommand{\urlprefix}{URL }
\providecommand{\doi}[1]{https://doi.org/#1}

\bibitem{angluinLearningRegularSets1987}
Angluin, D.: Learning regular sets from queries and counterexamples.
  Information and Computation  \textbf{75}(2),  87--106 (Nov 1987)

\bibitem{bergCorrespondenceConformanceTesting2005}
Berg, T., Grinchtein, O., Jonsson, B., Leucker, M., Raffelt, H., Steffen, B.:
  On the {{Correspondence Between Conformance Testing}} and {{Regular
  Inference}}. In: Cerioli, M. (ed.) Fundamental {{Approaches}} to {{Software
  Engineering}}. pp. 175--189. Lecture {{Notes}} in {{Computer Science}},
  {Springer}, {Berlin, Heidelberg} (2005)

\bibitem{broyModelbasedTestingReactive2005}
Broy, M. (ed.): Model-Based Testing of Reactive Systems: Advanced Lectures.
  No.~3472 in Lecture Notes in Computer Science, {Springer}, {Berlin ; New
  York} (2005)

\bibitem{clementeEfficientReductionNondeterministic2019a}
Clemente, L., Mayr, R.: Efficient reduction of nondeterministic automata with
  application to language inclusion testing. Logical Methods in Computer
  Science ; Volume 15 p. Issue 1 ; 18605974 (2019)

\bibitem{delahigueraGrammaticalInferenceLearning2010}
{De la Higuera}, C.: Grammatical Inference: Learning Automata and Grammars.
  {Cambridge University Press}, {Cambridge} (2010)

\bibitem{dorofeevaFSMbasedConformanceTesting2010a}
Dorofeeva, R., {El-Fakih}, K., Maag, S., Cavalli, A.R., Yevtushenko, N.:
  {{FSM-based}} conformance testing methods: {{A}} survey annotated with
  experimental evaluation. Information and Software Technology
  \textbf{52}(12),  1286--1297 (Dec 2010)

\bibitem{dorofeevaImprovedConformanceTesting2005}
Dorofeeva, R., {El-Fakih}, K., Yevtushenko, N.: An {{Improved Conformance
  Testing Method}}. In: Wang, F. (ed.) Formal {{Techniques}} for {{Networked}}
  and {{Distributed Systems}} - {{FORTE}} 2005. pp. 204--218. Lecture {{Notes}}
  in {{Computer Science}}, {Springer}, {Berlin, Heidelberg} (2005)

\bibitem{fujiwaraTestSelectionBased1991}
Fujiwara, S., {von Bochmann}, G., Khendek, F., Amalou, M., Ghedamsi, A.: Test
  {{Selection Based}} on {{Finite State Models}}. IEEE Transactions on Software
  Engineering  \textbf{17}(6),  591--603 (Jun 1991)

\bibitem{harrisSynthesisFiniteState1998}
Harris, M.: Synthesis of finite state machines: {{Functional}} optimization.
  Microelectronics Journal  \textbf{29}(6),  364--365 (Jun 1998)

\bibitem{learnlib}
Isberner, M., Howar, F., Steffen, B.: The open-source {{LearnLib}}. In:
  Kroening, D., P{\u a}s{\u a}reanu, C.S. (eds.) Computer Aided Verification.
  pp. 487--495. {Springer International Publishing}, {Cham} (2015)

\bibitem{jiangMinimalNFAProblems1993}
Jiang, T., Ravikumar, B.: Minimal {{NFA Problems}} are {{Hard}}. SIAM Journal
  on Computing  \textbf{22}(6),  1117--1141 (Dec 1993)

\bibitem{joonkikimSimplificationSequentialMachines1972}
{Joonki Kim}, Newborn, M.: The {{Simplification}} of {{Sequential Machines}}
  with {{Input Restrictions}}. IEEE Transactions on Computers
  \textbf{C-21}(12),  1440--1443 (Dec 1972)

\bibitem{kupfermanVerificationFairTransition1996}
Kupferman, O., Vardi, M.Y.: Verification of fair transition systems. In: Alur,
  R., Henzinger, T.A. (eds.) Computer {{Aided Verification}}. pp. 372--382.
  Lecture {{Notes}} in {{Computer Science}}, {Springer}, {Berlin, Heidelberg}
  (1996)

\bibitem{larrauriMinimizationSynthesisTail2021}
Larrauri, A., Bloem, R.: Minimization and {{Synthesis}} of the {{Tail}} in
  {{Sequential Compositions}} of {{Mealy Machines}}. arXiv:2105.10292 [cs]
  (Oct 2021)

\bibitem{mooreGedankenExperimentsSequentialMachines1956}
Moore, E.F.: Gedanken-{{Experiments}} on {{Sequential Machines}}. In: Shannon,
  C.E., McCarthy, J. (eds.) Automata {{Studies}}. ({{AM-34}}), pp. 129--154.
  {Princeton University Press} (Dec 1956)

\bibitem{peledBlackBoxChecking1999}
Peled, D., Vardi, M.Y., Yannakakis, M.: Black {{Box Checking}}. In: Wu, J.,
  Chanson, S.T., Gao, Q. (eds.) Formal {{Methods}} for {{Protocol Engineering}}
  and {{Distributed Systems}}, vol.~28, pp. 225--240. {Springer US}, {Boston,
  MA} (1999)

\bibitem{petrenkoTestingPartialDeterministic2005}
Petrenko, A., Yevtushenko, N.: Testing from partial deterministic {{FSM}}
  specifications. IEEE Transactions on Computers  \textbf{54}(9),  1154--1165
  (Sep 2005)

\bibitem{petrenkoTestingStrategiesCommunicating1995}
Petrenko, A., Yevtushenko, N., Dssouli, R.: Testing {{Strategies}} for
  {{Communicating FSMs}}. In: Mizuno, T., Higashino, T., Shiratori, N. (eds.)
  Protocol {{Test Systems}}: 7th Workshop 7th {{IFIP WG}} 6.1 International
  Workshop on Protocol Text Systems, pp. 193--208. {{IFIP}} \textemdash{} {{The
  International Federation}} for {{Information Processing}}, {Springer US},
  {Boston, MA} (1995)

\bibitem{petrenkoLearningCommunicatingState2019a}
Petrenko, A., Avellaneda, F.: Learning {{Communicating State Machines}}. In:
  Beyer, D., Keller, C. (eds.) Tests and {{Proofs}}. pp. 112--128. Lecture
  {{Notes}} in {{Computer Science}}, {Springer International Publishing},
  {Cham} (2019)

\bibitem{simaoReducingTestLength2012a}
Sim{\~a}o, A., Petrenko, A., Yevtushenko, N.: On reducing test length for
  {{FSMs}} with extra states. Software Testing, Verification and Reliability
  \textbf{22}(6),  435--454 (2012)

\bibitem{souchaFSMLib}
Soucha, M.: {{FSMLib}}. https://github.com/Soucha/FSMlib

\bibitem{souchaSPYHMethodImprovementTesting2018}
Soucha, M., Bogdanov, K.: {{SPYH-Method}}: {{An Improvement}} in {{Testing}} of
  {{Finite-State Machines}}. In: 2018 {{IEEE International Conference}} on
  {{Software Testing}}, {{Verification}} and {{Validation Workshops}}
  ({{ICSTW}}). pp. 194--203 (Apr 2018)

\bibitem{souchaObservationTreeApproach2020}
Soucha, M., Bogdanov, K.: Observation {{Tree Approach}}: {{Active Learning
  Relying}} on {{Testing}}. The Computer Journal  \textbf{63}(9),  1298--1310
  (Aug 2020)

\bibitem{souzaHSwitchCoverNew2017}
de~Souza, {\'E}.F., de~Santiago~J{\'u}nior, V.A., Vijaykumar, N.L.: H-{{Switch
  Cover}}: A new test criterion to generate test case from finite state
  machines. Software Quality Journal  \textbf{25}(2),  373--405 (Jun 2017)

\bibitem{t.s.chowTestingSoftwareDesign1978}
{T.S. Chow}: Testing {{Software Design Modeled}} by {{Finite-State Machines}}.
  IEEE Transactions on Software Engineering  \textbf{SE-4}(3),  178--187 (May
  1978)

\bibitem{vaandragerNewApproachActive2021a}
Vaandrager, F., Garhewal, B., Rot, J., Wi{\ss}mann, T.: A {{New Approach}} for
  {{Active Automata Learning Based}} on {{Apartness}}. arXiv:2107.05419 [cs]
  (Oct 2021)

\bibitem{vasilevskiiFailureDiagnosisAutomata1975}
Vasilevskii, M.P.: Failure diagnosis of automata. Cybernetics  \textbf{9}(4),
  653--665 (1975)

\bibitem{wangInputDonCare1993}
Wang, H., Brayton, R.K.: Input don't care sequences in {{FSM}} networks. In:
  Proceedings of 1993 {{International Conference}} on {{Computer Aided Design}}
  ({{ICCAD}}). pp. 321--328 (Nov 1993)

\end{thebibliography}
\newpage

\appendix
	
\section{Proof of Theorem~\ref{thm:seq}}
\label{ap:seq}
        It is clearly enough to show the result for $a=b$,
        $\alpha\models (s,a)\nsim (t,b)$ is equivalent to
        $\alpha\models (s,b)\nsim (t,b)$ whenever $\LA(a)\supseteq \LA(b)$. \par
		For any $s,t\in S_M$ and $a\in S_A$, we write $s\nsim_a t$ 
		as a shorthand for $(s,a)\nsim (t,a)$. Additionally,
		we say that $s\nsim_a^j t$ when
		$\alpha\models s\nsim_a t$ for some $\alpha$
		with $|\alpha|\leq j$.
		We adopt the convention that $s\sim_a^0 t$ for all 
		$s,t,a$. Let $m=|S_M||S_A|$.
		We show that $\sim_{a}^m$ is the same relation
		as $\sim_{a}$ for all $a\in S_A$. 
		Note that this proves our statement. 
		We proceed by showing various claims. The first ones are straight-forward.
		\par
		Claim 1. \textit{ The relation $\sim_a^{j+1}$
			refines $\sim_a^j$, written $\sim_a^j\supseteq \sim_a^{j+1}$. This means that $s \sim_a^{j+1} t$
			implies $s \sim_a^{j+1} t$ for all $s,t,a,j$. 
		}\par
		Claim 2. \textit{ The relations $\sim_a^j$, $\sim_a$ 
			are equivalence relations, and
			$\sim_a=\cap_{j=1}^\infty \sim_a^j$. }\par
		Claim 3. \textit{ Suppose that for some $j\in \N$
			it holds that $\sim_a^j=\sim_a^{j+1}$ for all $a\in S_A$. 
			Then $\sim_a^j=\sim_a^k$ for all $k\geq j$ and all $a\in S_A$.}
		To show this claim, suppose that $s\sim_a^j t$ but
		$s\nsim_a^{j+1} t$, for some $j>0$. Let $i_1i_2\dots i_{j+1} \in \Lcal(a)$ 
		be a sequence distinguishing $s$ and $t$. 
		Let $s^\prime\coloneqq\delta_M(s,i_1)$, $t^\prime
		\coloneqq \delta_N(t,i_1)$, and let $b\in \Delta_A(a,i_1)$ be such that 
		$i_2\dots i_{j+1}$ belongs to $\Lcal(b)$. 
		Then $s^\prime \nsim^{j}_b t^\prime$. Furthermore,  
		it cannot be that $s^\prime \nsim^{j-1}_b t^\prime$ 
		as well. Otherwise $s\nsim_a^j t$ would follow, 
		contradicting our initial assumption. Hence we have shown 
		that if for some $a\in S_A$, $j>0$ it holds
		$\sim_a^j\neq \sim_a^{j+1}$ then $\sim_b^{j-1} \neq 
		\sim_b^j$ for some $b\in S_B$. This is equivalent to the claim. \par
		
		Now we can complete the proof of our theorem. For each
		$j\in \N$ consider the set of equivalence relations $\{ \sim_a^j \}_{a\in S_A}$. Because of Claim 3, we know that at each successive step $j=1,2,\dots$ at least one
		relation is refined $\sim_a^j\supsetneq \sim_a^{j+1}$,
		until $\sim_a^j=\sim_a$ for all $a\in S_A$. If $\sim_a^j\supsetneq \sim_a^{j+1}$, then 
		$\sim_a^{j+1}$ yields strictly more 
		equivalence classes than $\sim_a^{j}$. For each $a$, 
		the relation $\sim_a^j$ can have at most $|S_M|$
		equivalence classes. Thus, the relations
		$\{ \sim_a^j \}_{a\in S_A}$ can be refined at most $m=|S_M||S_A|$ times in total. This implies $\sim_a^{m}=\sim_a$ for all $a$, 
		as we wanted to show. \qed

\section{Proof of
Theorem~\ref{thm:main_w_cover}}
\label{ap:thm_main}
We proceed by contradiction as in
Theorem~\ref{thm:main_no_cover}.
We take $N\in \Im_k$
satisfying $M\sim_E N$ and $M\nsim_\LA N$, and 
$\alpha\in \LA \setminus E$ a sequence
distinguishing $M$ and $N$ which minimizes $|\alpha|_\SubV$. 
This time we show that another sequence $\alpha^\prime$ separating $M$ and $N$ as well, with $|\alpha^\prime|_\SubV < |\alpha|_\SubV$ can be found. 
Let $\beta\in E$ be a $(k+1)$-redundant prefix of
$\alpha$, with $\beta\leq_\SubV \alpha$, 
and let $\gamma$ be the prefix for which
$\alpha=\beta \gamma$. Let $b\in \Delta_A(\beta)$
be a state satisfying $\gamma \in \LA(b)$. The node $\sfrac{b}{\beta}$ is $(k+1)$-redundant, so there
is a redundancy certificate $(R,B)$ witnessing this property. 
Let $R= (\sfrac{d_j}{\zeta_j})_{j=1}^\ell$.
As $|R\cup B|=k+1$, by the pigeonhole principle there are $\sfrac{c}{\varphi_1}
, \sfrac{\varphi_2}{c_2}\in C\cup B$ satisfying 
$\delta_N(\varphi_1)= \delta_N(\varphi_2)$.
Without loss of generality we can assume that $|\varphi_1|_\SubV \leq |\varphi_2|_\SubV$. 
As before, we proceed by giving various claims.

\begin{claim}[I] $(\delta_M(\varphi_1),c_1)\sim
	(\delta_M(\varphi_2),c_2)$.
\end{claim}
Otherwise we would have $\varphi_1 \divr \varphi_2$,
as $E$ is incompatibility preserving w.r.t. $R\cup B$. But this contradicts $M\sim_E N$, proving the claim.
\begin{claim}[II]
	Either one of the following holds. \textbf{Case 1:} $\sfrac{c_1}{\varphi_1}=\sfrac{d_x}{\zeta_x}$,
	$\sfrac{c_2}{\varphi_2}=\sfrac{d_y}{\zeta_y}$, for some indices
	$x<y$. ~ \textbf{Case 2:}
	$(\varphi_1, c_1)\in B$, $(\varphi_2, c_2)\in C$
\end{claim}
Last claim shows that both $\sfrac{c_1}{\varphi_1},
\sfrac{c_2}{\varphi_2}$ cannot belong to $B$ at the same time, as it would yield a conflict with the definition of basis. Thus there are two possible scenarios: either (i) both nodes belong to $R$, or (ii) exactly one of them lies in $B$. We show that these correspond to \textbf{Case 1} and \textbf{Case 2} in the statement, respectively. We begin by assuming (i).
In this situation, we know that  $\sfrac{c_1}{\varphi_1}=\sfrac{d_x}{\zeta_x}$,
$\sfrac{c_2}{\varphi_2}=\sfrac{d_y}{\zeta_y}$,
for some $x, y$, and we have to prove
$x<y$.
As $R\preceq_\SubV \sfrac{b}{\beta}$, 
it holds that
$\zeta_1\leq_\SubV \beta$, so there is
no $\xi\in V$ with $\zeta_1 < \xi < \beta$.
This implies $\zeta_1 <_\SubV \dots 
<_\SubV \zeta_\ell$, and as a consequence  $|\zeta_1|_\SubV < \dots < |\zeta_\ell|_\SubV$. By assumption $|\varphi_1|_\SubV \leq |\varphi_2|_\SubV$, 
so $x<y$ follows. Hence, \textbf{Case 1} holds. Now we assume (ii) instead. Note that for all $(\zeta, d)\in B$ it holds $\zeta\in V$, so $|\zeta|_\SubV=0$. Conversely, for all $(\zeta, d)\in C$, $\zeta\notin V$, and
$|\zeta|_\SubV>0$. Again, by assumption
$|\varphi_1|_\SubV \leq |\varphi_2|_\SubV$, implying $\sfrac{c_1}{\varphi_1}\in B$
and $\sfrac{c_2}{\varphi_2}\in C$, as in \textbf{Case 2}. \\~\\
For the remainder of the proof we will refer to the cases
 \textbf{Case 1} and \textbf{Case 2} in last 
 claim. 

\begin{claim}[III]
	The inequality $|\varphi_1|_\SubV \leq |\varphi_2|_\SubV$ is strict.
\end{claim}
\textbf{Case 1:} 
In the proof of Claim (II) we showed $|\zeta_1|_\SubV < \dots < |\zeta_\ell|_\SubV$. Hence the statement follows. In this situation
$\varphi_1, \varphi_2$ lie among the
$\zeta_j$'s, so the ranking of inequalities implies our claim. ~ 
\textbf{Case 2:} Note that for all
$(\zeta, d)\in B$ it holds $\zeta\in V$, so
$|\zeta|_\SubV=0$. Conversely, for all
$(\zeta, d)\in C$, $\zeta\notin V$, and
$|\zeta|_\SubV>0$. This shows the claim. 
\begin{claim}[IV] $c_1\sqsupseteq c_2$.
\end{claim}
\textbf{Case 1:} The statement follows
from the definition of monotonous ranking. ~ 
\textbf{Case 2:} By \textit{Claim (III)}, it holds
$(\varphi_1, c_1)\in B$,
$(\varphi_2, c_2)\in C$. Using the definition of basis and \textit{Claim (I)},
we obtain $c_1\sqsupseteq c_2$ in this case as well. \\~\\
Now we are in conditions to build the
second distinguishing sequence $\alpha^\prime$. By \textit{Claim (II)} 
$\sfrac{c_2}{\varphi_2}\in C$, so $\varphi_2\leq \beta$. Let $\omega$ be the suffix satisfying $\beta= \varphi_2\omega$. Then
$\alpha= \varphi_2\omega\gamma$. We
define $\alpha^\prime$ as the word
$\varphi_1\omega\gamma$. 
They following claims can all be shown exactly as in Theorem~\ref{thm:main_no_cover}'s proof:
\textit{Claim (V).} $\alpha^\prime \in \LA$. ~ \textit{Claim (VI).}
$\lambda_M(\delta_M(\varphi_{2}),\omega\gamma)\neq \lambda_N(\delta_N(\varphi_{2}),\omega\gamma)$. ~
\textit{Claim (VII).} $\lambda_N(\delta_N(\varphi_{1}),\omega\gamma)= \lambda_N(\delta_N(\varphi_{1}),\omega\gamma)$. ~ \textit{Claim (VIII).}
$\lambda_M(\delta_M(\varphi_{1}),\omega\gamma)= \lambda_M(\delta_M(\varphi_{2}),\omega\gamma)$.\\~\\
Claims \textit{(V)-(VII)} show that $\alpha^\prime$ belongs to $\LA$ and distinguishes $M$ from $N$. All that is 
left is to prove $|\alpha^\prime|_\SubV < |\alpha|_\SubV$. Using 
$|\alpha|_\SubV = \varphi_1$
and $\alpha^\prime= \varphi_2 \omega \gamma$
we obtain
(1) $|\alpha^\prime|_\SubV \leq |\varphi_2|_\SubV
+ |\omega\gamma|$.
Now we show a similar expression for 
$|\alpha|_\SubV$.
As $(\varphi_2, c_2)\preceq_\SubV \sfrac{b}{\beta}$ it holds
$\varphi_2 \leq_\SubV \beta$. Also, by hypothesis, $\beta \leq_\SubV \alpha$. Putting the inequalities together we get
$\varphi_2\leq_\SubV \alpha$. This
yields $|\alpha|_\SubV=|\varphi_2|_\SubV +  |\omega \gamma|$. Additionally
$|\alpha^\prime|_\SubV \leq |\varphi_1|_\SubV
+ |\omega\gamma|$. Comparing the expression for $\alpha$ and $\alpha^\prime$ and utilizing \textit{Claim (III)} gets us $|\alpha^\prime|_\SubV
< |\alpha|_\SubV$. This contradicts our initial choice of $\alpha$ and completes the proof of the theorem. \qed
\par

\section{Detailed Description of Complex} \label{ap:complex}
	Algorithm~\ref{alg:second}
	shows the main structure of 
	\textsc{Complex}.
	The routine $\textsc{Core}()$ first obtains a core
	$Q$ from $\textsc{WeakCore}$ and afterwards removes each 
	location $(s,a)\in Q$ if there is another one $(t,b)\in Q$ where $(s,a)\sim (t,b)$
	and $b\sqsupseteq a$. The second difference
	is that \textsc{Complex} does not make use of harmonized quasi-identifiers, unlike \textsc{Simple}, but instead relies on a map of
	distinguishing sequences $DistSeqs$. This map stores a shortest 
	separating sequence $\alpha\models s\nsim_a t$ for each triple $s,t\in S_M$, $a\in S_A$, if it exists, or the empty sequence otherwise. Finally, the last difference is that in \textsc{Complex} no distinguishing sequences are added
	to the cover $V$ initially. Instead we add these sequences dynamically 
	during the exploration process. \par
	
		\begin{algorithm}
		\caption{\textsc{Complex}($M,A,\sqsubseteq, k$)}\label{alg:second}
		\hspace*{\algorithmicindent} \textbf{Input}
		A specification machine $M$, context automaton $A$, 
		an under-approximation $\sqsubseteq$ of language containment over $S_A\times S_A$,
		and a bound $k$. \\
		\hspace*{\algorithmicindent} \textbf{Output} A $k$-complete
		suite $E$ for $M$ in the context of $A$.
		\begin{algorithmic}[1]
			\State $Q\gets \textsc{Core}()$
			\State $V, \, toCvr \gets \textsc{Cover}(Q)$
			\Comment{
				\parbox[t]{.4\linewidth}{
					$toCvr$ is a map $Q \rightarrow V$ where $(s,a)\in \Delta_{M\times A}(toCvr(s,a))$
				}
			} 
			\State $E\gets V$
			\State $DistSeqs \gets$ map $S_M\times S_M \times A \rightarrow I_M^*$ assigning a shortest distinguishing sequence $\alpha \models
			s\nsim_a t$ for each $s,t\in S_M$, $a\in S_A$. 
			\State \textbf{for all} $\alpha \in V$ \textbf{do}
			$\alpha_\SubV\gets \alpha$, and $\textsc{Explore}(\emptyword)$ 
			\State \Return $E$
		\end{algorithmic}
	\end{algorithm}

	In its final step, \textsc{Complex} performs a depth-first
	search from each word $\alpha_\SubV\in V$, enlarging $E$ along the
	way. For this, the algorithm relies on the same routines
 	\textsc{Explore} and \textsc{SearchCerts} utilized by \textsc{Simple}.
 	We modify, however the functions \textsc{BuildRankings}, \textsc{Basis},
 	and \textsc{ExploitCert}. Now we can use $\sqsupseteq$ to produce 
 	general monotonous rankings, instead of only flat ones as before. This allows \textsc{Complex} to potentially prune the search space earlier,
 	as it can force shorter sequences to become $(k+1)$-redundant.\par
 	Similarly as with \textsc{Simple}, the method \textsc{BuildRankings}$(\beta,a)$ builds a family $Rankings$ of 
 	monotonous rankings $R\preceq_\SubV (\alpha_\SubV \beta,a)$. It does so  	
 	by building for each $b\in S_A$ a maximum-length ranking
 	$(\sfrac{c_j}{\varphi_j})_{j=1}^\ell$ where $c_j=b$. This can be
 	done incrementally by scanning the nodes 
 	$\sfrac{c}{\alpha_\SubV \beta^\prime}\preceq_\SubV
 	\sfrac{a}{\alpha_\SubV \beta}$
 	for each prefix $\beta^\prime \leq \beta$.

	\begin{algorithm}[tbh]
	\caption{\textsc{BuildRankings},
		\textsc{Basis}, \textsc{ExploitCert}
		\hfill (Complex version)} \label{algs:advanced_version}
	\begin{algorithmic}[1]
		\Procedure{BuildRankings}{$\beta,a$}
		\Statex \textbf{Input} a suffix $\beta$ with
		$\alpha_\SubV\beta\in \LA$, and a state
		$a\in \Delta_A(\alpha_\SubV\beta)$ 
		\Statex \textbf{Output} A set $Rankings$ of
		monotonous rankings $R\preceq_\SubV \frac{a}{\alpha_\SubV\beta}$.
		\State $Rankings \gets \{ \}$
		\State initialize empty rankings
		$R^0_{b_1}, R^0_{b_2}, \dots$ for all 
		$b_i\in S_A$.
		\State $\Omega \gets $ set of nodes
		$\sfrac{c}{\varphi} \preceq_\SubV \sfrac{a}{\alpha_\SubV\beta}$.
		\ForAll{ $j = 1,2, \dots, |\beta|$, and \textbf{all} $b\in S_A$}
		\If{$\sfrac{b}{\alpha_\SubV\cdot\beta_{\leq j}}\in \Omega$}
		\State Let $c\in S_A$ be the state $c\sqsupseteq b$ maximizing
		$|R^{j-1}_c|$. 
		\State $R^j_b =  R^{j-1}_c \cup \frac{b}{\alpha_\SubV\cdot\beta_{\leq i}}$
		\Else \, \, $R^j_b = R^{j-1}_b$.
		\EndIf
		\EndFor
		\State $Rankings\gets \{R^{|\beta|}_{b}\}_{b\in S_A}$
		\EndProcedure
		\Statex
		\Procedure{Basis}{$R$}
		\Statex	\textbf{Input} A monotonous ranking $R=(\sfrac{c_j}{\varphi_j})_{j=1}^\ell \subseteq \Gamma$.
		\Statex \textbf{Output} A basis $B$ for $R$.
		\State $B \gets \{ \}$
		\State $Q^\prime \gets \{ \}$
		\ForAll{ $j= 1,2,\dots, \ell$ and $s\in S_M$}
		\If{ $Q^\prime$ does not contain any $(t,b)$ with $(t,b)\sqsupseteq (s,c_j)$ }
		\State Find $(t,b)\in Q$ with $(t,b)\sqsupseteq (s,c_j)$, and 
		add it to $Q^\prime$.
		\EndIf
		\EndFor
		\State \textbf{for all }$(s,c) \in Q^\prime$, add $(toCvr(s,c),c)$ to $B$
		\State \Return $B$
		\EndProcedure
		\Statex
		\Procedure{ExploitCert}{$R,B$}
		\Statex A redundancy certificate $(R,B)$, where $R=(\sfrac{c_j}{\varphi_j})_{j=1}^\ell$.
		\ForAll{$j=1,\dots,\ell$ and \textbf{all } $(\omega,b)\in B$}
		\State $s\gets \delta_M(\varphi_j)$, $t\gets \delta_M(\omega)$
		\State $x \gets$ maximum index $j\leq x \leq \ell$ satisfying 
		$s \nsim_{c_x} t$.
		\State Add $\varphi_j \gamma$ to $E$, where $\gamma= DistSeqs(s,t,c_x)$  
		\EndFor
		\ForAll{pairs $(\omega_1,b_1),(\omega_2,b_2)\in B$}
		\State $s_1\gets \delta_M(\omega_1)$, $s_2\gets \delta_M(\omega_2)$
		\State  $x \gets$ maximum index $1\leq x \leq \ell$ satisfying 
		$s_1 \nsim_{c_x} s_2$.
		\State Add $\omega_1\gamma, \omega_2\gamma$ to $E$, where
		$\gamma= DistSeqs(s_1,s_2,c_x)$
		\EndFor
		\EndProcedure			
	\end{algorithmic}
	\end{algorithm}

	Finding a greatest basis $B\subseteq \Gamma(V)$ for a 
	monotonous ranking is, in principle, computationally hard, 
	given that this task can be reduced to 
	a maximal independent set problem. However, if we do not aim 
	for a biggest basis, the task can be carried out
	with relative efficiency. We propose a greedy approach
	in \textsc{Basis} for this purpose. \par
	Finally, \textsc{ExploitCert}$(R,B)$ is tasked with making
	$E$ incompatibility preserving w.r.t. the certificate $(R,B)$
	by adding various distinguishing sequences, as before. 
	Following a naive approach involves adding a distinguishing sequence for 
	each pair $(\omega_1, b_1),(\omega_2,b_2)\in C\cup B$ where
	$(\delta_M(\omega_1),b_1)\nsim (\delta_M(\omega_2),b_2)$. This adds up
	to potentially $(|C|+|B|)^2=(k+1)^2$ sequences.
	However, if one chooses the sequences carefully, it is only needed to distinguish
	the nodes in $R$ with those in $B$, and the nodes in $B$ among themselves.
	As $|B|\leq |S_M|$, this brings down the number
	of separating sequences to at most $(k+1)|S_M|$.
	We implement this technique in \textsc{ExploitCert}.

\section{Proof of Theorem~\ref{thm:length}}
\label{ap:length}
	We show that in the body of the main loop \textsc{SearchCerts} (Algorithm~\ref{alg:search_certs})
	is able to find a big enough redundancy certificate for all nodes $\sfrac{a}{\alpha_V\beta}$. 
	For each $a\in S_A$ let $m_a$ be the number of different classes
	in $\sfrac{S_{M\times A}}{\cong}$ corresponding to locations of the form
	$(s,a)$.
	Then $|\beta|=(k|S_A| - \sum_{a\in S_A} m_a)  + 1$. 
	Let $a\in \Delta_A(\alpha_V\beta)$. Then there is at least one
	sequence of nodes $(\alpha_V\beta_{\leq 1}, b_1)\preceq 
	(\alpha_V\beta_{\leq 2}, b_2) \preceq  \dots
	\preceq (\alpha_V\beta, a)$. This sequence
	has length $|B|$, so by the pigeonhole principle
	at one state $b\in S_A$ occurs at least 
	$k - m_b + 1$ times throughout the succession. 
	Thus, this 
	quantity is a lower bound for the size of
	the flat ranking $R_b$, corresponding to $b$, built in the procedure 
	\textsc{BuildRankings}$(\beta, a)$.
	Now, note that \textsc{Basis}$(R_b)$ returns
	a basis of size exactly $n_b$. Hence,
	$|R_b| + |\textsc{Basis}(R_b)|\geq k+1$, 
	and the conditional in Algorithm~\ref{alg:search_certs},
	line~\ref{lin:certs_if} is true. Our initial
	choice of $a\in \Delta_A(\alpha_V\beta)$ was arbitrary,
	so this proves that $\textsc{SearchCerts}(\beta)$ 
	does not return $false$. \qed

	\section{Time-Cost Analysis for
	\textsc{Simple} and \textsc{Complex}}
	\label{ap:cost}
	
	For the complexity analysis of $\textsc{Simple}$, we only need to consider the time spent in the
    routine $\textsc{Explore}$. Potentially, this
    function is called once for each word $\alpha_\SubV \beta$, where
    $\alpha_\SubV\in V$, and $|\beta|\leq |S_A|k - n_{M\times A} + 1$.
    During these calls, $\textsc{Explore}$ invokes the 
	\textsc{SearchCerts} once, on $\beta$, and the method
	\textsc{ExploitCert} at most $|S_A|$ times: one for each
	certificate returned by \textsc{SearchCerts}. We analyze both 
	functions separately. \par
	Inside \textsc{SearchCerts}$(\beta )$ most of the time
	is spent calling \textsc{BuildRankings}$(\beta, a)$. 
	In this second function the bulk of the time is invested in
	building the set $\Omega$ of nodes 
	$\sfrac{c}{\varphi}\preceq_\SubV (\alpha_V\beta, a)$. This can be done
	by back-propagating the node
	$(\alpha_V\beta, a)$ throughout all words $\alpha_V\beta^\prime$
	with $\beta^\prime\leq\beta$. If one stores $A$ reverse transitions,
	this takes at most $O(|S_A|^2|\beta|)$ time. \textsc{BuildRankings} is called
	at most $|S_A|$ times in the a single
	call of \textsc{SearchCerts}. Hence, \textsc{SearchCerts}$(\beta)$
	takes $O(|S_A|^3|\beta|)=
	O(|S_A|^3(e))$ time. The method is called once for
	each word $\alpha_V\beta$, so the total amount of time
	it uses during \textsc{Simple} is 
	$O(|S_A|^4|S_M|(e)|I_M|^{e+1})$ time.
	\par
	The workload inside \textsc{ExploitCert}$(R,B)$ is mainly
	the result of adding distinguishing suffixes. This method adds
	at most $|R||S_M|$ of those to $E$, and each one of these
	sequences has length bounded by $|S_M||S_A|$. If $E$ is stored
	in a tree structure, this can be done in $O(|C||S_A||S_M|^2)=
	O(k|S_A||S_M|^2)$ time. The method is called at most 
	$|S_A|$ times for each word $\alpha_V\beta$. Hence,
	\textsc{Simple} spends at most $O(k|S_A|^3|S_M|^3(e)|I_M|^{e+1})$ time in \textsc{ExploitCert}.
	Putting the bounds for \textsc{SearchCerts} and \textsc{ExploitCert}
	together gives us that the total time cost of \textsc{Simple}$(M,A,k)$
	is $O((k|S_A|^3|S_M|^3 + |S_A|^4|S_M|)e|I_M|^{e+1})$. Analogous arguments can be used
	to obtain the complexity of \textsc{Complex}.
	The only relevant change here is that there is an additional
	inner loop in the routine \textsc{BuildRankings}, increasing its cost
	by a factor of $|S_A|$. This yields a total complexity of 
	$O((k|S_A|^3|S_M|^3 + |S_A|^5|S_M|)e|I_M|^{e+1})$ for 
	\textsc{Complex}$(M,A,\sqsupseteq,k)$.

\end{document}